\documentclass[superscriptaddress, twocolumn, amsmath, amssymb, aps,
pra, letterpaper, notitlepage]{revtex4-1}

\usepackage{}

\usepackage{graphicx,graphics,epsfig,subfigure,times,bm,bbm,amssymb,amsmath,amsfonts,amsthm,mathrsfs,MnSymbol}
\usepackage[matrix,frame,arrow]{xypic}
\usepackage[normalem]{ulem}
\usepackage{slashed}
\usepackage{stmaryrd}
\usepackage{color}
\usepackage[usenames,dvipsnames,svgnames,table]{xcolor}
\definecolor{darkblue}{rgb}{0.0,0.0,0.3}
\usepackage[colorlinks=true,
            linkcolor=red,
            urlcolor= darkblue,
            citecolor=blue]{hyperref}

\usepackage{enumerate}
\usepackage{epstopdf}
\usepackage{relsize}

\newtheorem{definitionenv}{Definition}
\newtheorem{remarkenv}[definitionenv]{Remark}
\newenvironment{remark}{\begin{remarkenv}\rm}{\end{remarkenv}}

\newtheorem{exampleenv}{Example}

\newtheorem{mydef}{Definition}

\newtheorem{mytheorem}{Theorem}
\newtheorem{mylemma}{Lemma}

\newcommand{\mcal}{\mathcal}

\newcommand{\bes} {\begin{subequations}}
\newcommand{\ees} {\end{subequations}}
\newcommand{\bea} {\begin{eqnarray}}
\newcommand{\eea} {\end{eqnarray}}

\newcommand{\red}[1]{\textcolor{red}{#1}}

\newcommand{\beq}{\begin{equation}}
\newcommand{\beqs}{\begin{equation*}}

\newcommand{\eeq}{\end{equation}}
\newcommand{\eeqs}{\end{equation*}}

\newcommand{\ignore}[1]{}

\def\>{\rangle}
\def\<{\langle}

\def\Pr{\mathrm{Pr}}
\def\llb{\llbracket}
\def\rrb{\rrbracket}

\newcommand{\ket}[1]{|#1\rangle}

\newcommand{\sfH}{\textsf{H}}

\newcommand{\cC}{\mathcal{C}}

\newcommand{\cG}{\mathcal{G}}

\begin{document}
\title{Constant depth fault-tolerant Clifford circuits for multi-qubit large block codes}
\author{Yi-Cong Zheng}

\email{
yicongzheng@tencent.com}
\affiliation{Tencent Quantum Lab, Tencent, Shenzhen, Guangdong, China, 518057}
\affiliation{Centre for Quantum Technologies, National University of Singapore, Singapore 117543}
\affiliation
{Yale-NUS College, Singapore 138527}

\author{Ching-Yi Lai}
\affiliation{Institute of Communications Engineering, National Chiao Tung University, Hsinchu 30010, Taiwan}
\author{Todd A. Brun}
\affiliation{Ming Hsieh Department of Electrical and Computer Engineering, University of Southern California, Los Angeles, California 90089, USA\\}
\author{Leong-Chuan Kwek}
\affiliation{Centre for Quantum Technologies, National University of Singapore, Singapore 117543}
\affiliation{MajuLab, CNRS-UNS-NUS-NTU International Joint Research Unit, UMI 3654, Singapore}
\affiliation{
Institute of Advanced Studies, Nanyang Technological University, Singapore 639673}
\affiliation{National Institute of Education, Nanyang Technological University, Singapore 637616 }
\date{\today}

\begin{abstract}
Fault-tolerant quantum computation (FTQC) schemes using large block codes that encode $k>1$ qubits in $n$ physical qubits can potentially reduce the resource overhead to a great extent because of their high encoding rate. However, the fault-tolerant (FT) logical operations for the encoded qubits are difficult to find and implement, which usually takes not only a very large resource overhead but also long \emph{in-situ} computation time. In this paper, we focus on Calderbank-Shor-Steane $\llb n,k,d\rrb$ (CSS) codes and their logical FT Clifford circuits. We show that the depth of an arbitrary logical Clifford circuit can be implemented fault-tolerantly in $O(1)$ steps \emph{in-situ} via either Knill or Steane syndrome measurement circuit, with the qualified ancilla states efficiently prepared. Particularly, for those codes satisfying $k/n\sim \Theta(1)$, the resource scaling for Clifford circuits implementation on the logical level can be the same as on the physical level up to a constant, which is independent of code distance $d$. With a suitable pipeline to produce ancilla states, our scheme requires only a modest resource cost in physical qubits, physical gates, and computation time for very large scale FTQC.
\end{abstract}
\maketitle

\section{Introduction}
Quantum error-correcting codes (QECCs)~\cite{Shor:1995:R2493,Steane:1996:793,Calderbank:1996:1098,Gaitan:2008:CRC,QECbook:2013}
and the theory of fault-tolerant quantum computation (FTQC)~\cite{Shor:1996:56,Aharonov:1997:176,
Gottesman:9705052,Kitaev:2003:2,
DivencenzoFTPhysRevLett.77.3260, KnillFTNature,Aharonov:2006:050504,QECbook:2013} have shown that large-scale quantum computation is possible if the noise is not strongly correlated between qubits and its rate is below certain threshold~\cite{Aharonov:1997:176,KnillFTNature,
Terhal:2005:012336,Aharonov:2006:050504,Aliferis:2006:97,cross2007comparative,
Aliferis:2008:181}.

Large QECCs with high encoding rates typically encode many logical qubits with high distance. FTQC architectures based on these codes may potentially outperform smaller codes and topological codes, like surface codes~\cite{Kitaev:2003:2, Folwer2012PhysRevA.86.032324} and color codes ~\cite{Bombin:2006:180501}, in terms of the overall resource required and the error correction ability~\cite{steane1999efficient_Nature, steane2005fault, brun2015teleportation, Steane:2003:042322, gottesman2013overhead}. However, for an $\llb n,k,d\rrb$ code with $k,d\gg 1$, it may be extremely difficult (or even impossible) to find all required fault-tolerant (FT) logical gates. For Calderbank-Shor-Steane (CSS) codes~\cite{Calderbank:1996:1098,Steane:1996:793}, one way to resolve this challenge is to implement logical circuits indirectly through Knill or Steane syndrome extraction circuits~\cite{KnillFTNature, steane1997active} with additional blocks of encoded ancilla qubits prepared in specific states~\cite{steane1997active, Gottesman:1999:390,Zhou:2000:052316, brun2015teleportation}.
Unfortunately, the distillation processes for each encoded ancilla state are complicated, and different ancilla states are usually required for each logical gate. As an example, a Clifford circuit on $k$ qubits requires $O(k^2/\log k)$ Clifford gates~\cite{markov2008optimal, aaronson2004improved} with circuit depth $O(k)$; if an $\llb n,k, d \rrb$ CSS code is used, it requires $O(k^2/\log k)$ logical Clifford gates~\cite{aaronson2004improved}, and in general, $O(k^2/\log k)$ different ancilla states need to be prepared, and the same number of Knill/Steane syndrome extraction steps are required.

A natural question arises: can one implement logical circuits on those multi-qubit large block codes ($k\gg 1$) in a quicker and more efficient way? In this paper, we show that for Clifford circuits, the answer is positive for CSS codes: one can implement an arbitrary logical Clifford circuit fault-tolerantly using $O(1)$ qualified encoded ancilla states and a constant number of Knill/Steane syndrome measurement steps. Thus the depth of a logical Clifford circuit can be reduced to $O(1)$ \emph{in-situ}. Furthermore, we show that with the distillation protocol proposed in ~\cite{Ancilla_distillation_1,zheng2017efficient}, these ancilla states can be distilled \emph{off-line} in ancilla factories with yield rate close to $O(1)$ asymptotically, if the physical error rate is sufficiently low. Especially, for those families of large block codes with $k/n\sim \Theta(1)$, the number of physical qubits and physical gates required for an arbitrary logical Clifford circuit can scale as $O(k)$ and $O(k^2/\log k)$ respectively on average. These results suggest that the resource cost of Clifford circuits on the logical level can scale the same as on the physical level, if the distillation circuits and large block quantum codes are carefully chosen. With a proper pipeline structure of ancilla factories to work in parallel, we are also convinced that the scaling of the required resources including the overall number of qubits, physical gates and the computation time, can be very modest for large scale FTQC.

The structure of the paper is as follows. We review preliminaries and set up notation in Sec.~\ref{sec:prim}.
In Sec.~\ref{sec:constant_depth}, we propose our scheme to implement FT logical Clifford circuits via a constant number of Knill or Steane syndrome measurement. The resource overhead for the scheme is carefully analyzed.
In Sec.~\ref{sec:discussion}, we compare our scheme to some other closely-related FTQC schemes according to the resource overhead and real-time computational circuit depth.


\section{PRELIMINARIES and notation}\label{sec:prim}
\subsection{Stabilizer formalism and CSS codes}
Let $\mathcal{P}_n=\mathcal{P}_1^{\otimes n}$ denote the $n$-fold Pauli group, where
\begin{equation*}
\mathcal{P}_1=\{\pm I, \pm i I, \pm X, \pm i X, \pm Y, \pm i Y, \pm Z, \pm i Z\},
\end{equation*}
and $I={\footnotesize \left(
         \begin{array}{cc}
           1 & 0 \\
           0 & 1 \\
         \end{array}
       \right)}$, $X={\footnotesize \left(
         \begin{array}{cc}
           0 & 1 \\
           1 & 0 \\
         \end{array}
       \right)}
$, $Z={\footnotesize \left(
         \begin{array}{cc}
           1 & 0 \\
           0 & -1 \\
         \end{array}
       \right)}$, and $Y=iXZ$ are the Pauli matrices.

Let $X_j$, $Y_j$, and $Z_j$ act as single-qubit Pauli matrices on the $j$th qubit and trivially elsewhere.
We also introduce the notation $X^{\mathbf a}$, for ${\mathbf a}=a_1\cdots a_n\in \mathbb{Z}_2^n$, to denote the operator $\otimes_{j=1}^n X^{a_j}$ and let $\text{supp}({\mathbf a})=\{j:a_j=1\}$.
For ${\mathbf a}, {\mathbf b}\in \mathbb{Z}_2^n$, define $\mathcal{I}_{\bf ab}=\text{supp}({\bf a})\bigcap\text{supp}({\bf b})$ and let $\tau_{\bf ab}=\left|\mathcal{I}_{\bf ab}\right|$ be the size of $\mathcal{I}_{{\bf a}{\bf b}}$.  An $n$-fold Pauli operator can be expressed as
\begin{equation}\label{eq:general_error}
i^l\cdot \bigotimes_{j=1}^n X^{a_j}Z^{b_j}=i^l X^{\bf a}Z^{\bf b}, \quad {\bf a},{\bf b}\in \mathbb{Z}^n_2, \ l\in\{0,1,2,3\}.
\end{equation}
Then $({\bf a}\,|\,{\bf b})$ is called the
\emph{binary representation} of the Pauli operator $i^lX^{\bf a}Z^{\bf b}$ up to an overall phase $i^l$. In particular, $\pm i^{\tau_{\bf ab}} X^{\bf a}Z^{\bf b}$, which is Hermitian, has eigenvalues $\pm 1$. From now on we use the binary representation and neglect the overall phase for simplicity when there is no ambiguity. We define the weight of $E$, $\text{wt}(E)$, as the number of terms in the tensor product which are not equal to the identity.

Suppose $\mathcal{S}$ is an Abelian subgroup of $\mathcal{P}_n$
with a set of $n-k$ independent and commuting generators $\{S_1=i^{\tau_{{\bf a}_1{\bf b}_1}} X^{{\bf a}_1}Z^{{\bf b}_1},\dots, S_{n-k}=i^{\tau_{{\bf a}_{n-k}{\bf b}_{n-k}}} X^{{\bf a}_{n-k}}Z^{{\bf b}_{n-k}}\}$, and $\mathcal{S}$ does not include $-I^{\otimes n}$. An $\llb n,k\rrb$ quantum stabilizer code $C(\mathcal{S})$
is defined as the $2^{k}$-dimensional subspace of the $n$-qubit state space ($\mathbb{C}^{2^n}$) fixed by  $\mathcal{S}$,
which is the joint $+1$ eigenspace of $S_1, \dots, S_{n-k}$.
Then for a codeword $\ket{\psi}\in C(\mcal{S})$, $$S\ket{\psi}=\ket{\psi}$$ for all $S\in \mathcal{S}$. We also define $N(\mathcal{S})$ to be the normalizer of the stabilizer group. Thus any non-trivial logical Pauli operator on codewords belongs to $N(\mathcal{S})\backslash\mathcal{S}$ and let $X_{j,L}$, $Y_{j,L}$ and $Z_{j,L}$ be logical Pauli operators acting on the $j$th logical qubit. The distance $d$ of the code is defined as
$$d=\min_{L\in N(\mathcal{S})\backslash \mathcal{S}} \text{wt}(L).$$ Suppose $\mathcal{S'}\in \mathcal{P}_n$ is another Abelian subgroup containing $\mathcal{S}$ with $k=0$, then $C(\mathcal{S}')$ has only one state $|\psi\>$ up to a global phase. This state is called a \emph{stabilizer codeword} of $\mathcal{S}$, whose binary representation is
$$
\psi = \left(
         \begin{array}{c|c}
          {\bf a}_1 & {\bf b}_1 \\
          \vdots & \vdots \\
          {\bf a}_{n} & {\bf b}_{n}\\
         \end{array}
       \right).
$$

If a Pauli error $E$ corrupts $\ket{\psi}$, some eigenvalues of $S_1,\dots, S_{n-k}$ may be flipped, if they are measured on $E|\psi\>$.
Consequently, we gain information about the error by measuring the stabilizer generators $S_1,\dots, S_{n-k}$,
and the corresponding measurement outcomes (in bits) are called the \emph{error syndrome} of $E$. A quantum decoder has to choose a good recovery operation based on the measured error syndromes.

CSS codes are an important class of stabilizer codes for FTQC. Their generators are tensor products of the identity and either $X$ or $Z$ operators (but not both)~\cite{Calderbank:1996:1098, Steane:1996:793}. More formally, consider two classical codes, $\mathcal{C}_Z$ and $\mathcal{C}_X$ with parameters $[n, k_Z, d_Z]$ and $[n, k_X, d_X]$, respectively, such that $\mathcal{C}_X^\perp \subset \mathcal{C}_Z$. The corresponding parity-check matrices are $\sfH_Z$ ($(n-k_Z)\times n$) and $\sfH_X$ ($(n-k_X)\times n$) with full rank.
One can form an $\llb n, k=k_X+k_Z-n, d \rrb$ CSS code, where $d\geq\min\{d_Z,d_X\}$. In general, a logical basis state can be represented as:
$$
|u\>_L=\sum_{x\in \mathcal{C}_{X}^\perp}|x+u  D\>,
$$
where $u\in \mathbb{Z}_2^k$ and $D$ is a $k\times n$ binary  matrix, whose rows are the coset leaders of $\mathcal{C}_Z/\mathcal{C}_X^\perp$. The stabilizer generators of a CSS code in binary representation are:
$$
\left(
  \begin{array}{c|c}
    \textsf{H}_Z & {\bf 0} \\
    {\bf 0} & \textsf{H}_X \\
  \end{array}
\right),
$$
where $\textsf{H}_X(\textsf{H}_Z)$ is made of $Z(X)$ type Pauli operators. For the special case that $\cC_{X}=\cC_{Z}$, we call such a code self-dual CSS code.

\subsection{Clifford circuits}\label{sec:stabilizer_circuit}
Clifford circuits are composed solely of Hadamard ({\rm H}), Phase ({\rm P}), and controlled-NOT ({\rm CNOT}) gates, defined as
\begin{equation*}
\text{H} = \frac{1}{\sqrt{2}}\left(
             \begin{array}{cc}
               1 & 1 \\
               1 & -1 \\
             \end{array}
           \right),\
\text{P} = \left(
             \begin{array}{cc}
               1 & 0 \\
               0 & i \\
             \end{array}
           \right),\
{\small \text{CNOT} = \left(
                \begin{array}{cccc}
                  1 & 0 & 0 & 0 \\
                  0 & 1 & 0 & 0 \\
                  0 & 0 & 0 & 1 \\
                  0 & 0 & 1 & 0 \\
                \end{array}
              \right).}
\end{equation*}
The $n$-qubit Clifford circuits form a finite group, which, up to overall phases, is isomorphic to the binary symplectic matrix group defined in~\cite{aaronson2004improved}:
\begin{mydef}[Symplectic group]\label{def:symplectic_group}
The group of $2n\times 2n$ symplectic matrices over $\mathbb{Z}_2$ is defined in:
\beqs
{\rm Sp}(2n, \mathbb{Z}_2)\equiv \{M\in {\rm GL}(2n, \mathbb{Z}_2): MJ_nM^t=J_n\}
\eeqs
under matrix multiplication. Here
$
J_n=\left(
             \begin{array}{c|c}
               {\bf 0} & I_n \\
               I_n & {\bf 0} \\
             \end{array}
           \right).
$
\end{mydef}
In general, $M\in {\rm Sp}(2n, \mathbb{Z}_2)$ has the form
$$
M=\left(
  \begin{array}{c|c}
    Q & R \\
    \hline
    \rule[0.4ex]{0pt}{8pt}
    S & T \\
  \end{array}
\right),  \label{eq:symplectic_matrix}
$$
where $Q$, $R$, $S$ and $T$ are $n\times n$ square matrices satisfying the following conditions:
$$
QR^t=RQ^t, \quad ST^t = TS^t, \quad Q^tT + R^t S = I_n.
$$
In other words, the rows of $(Q \ | \ R)$ are symplectic partners of the rows of $(S \ | \ T)$.
Thus, an  $n$-qubit Clifford circuit can be represented by a $2n\times 2n$ binary matrix with respect to the basis of the binary representation of Pauli operators in (\ref{eq:general_error}). Then $UX^{{\bf a}}Z^{{\bf b}}U^\dag$ is represented by $({\bf a},{\bf b})M_U$, where $M_U$ is the binary symplectic matrix corresponding to $U$.
For example, the idle circuit (no quantum gates) is represented by $I_{2n}$, the $2n\times 2n$  identity matrix.
The representation of consecutive Clifford circuits $M_1,\dots, M_j$  is their binary matrix product
$$
M=M_1\cdots M_j.
$$
We emphasize here that the symplectic matrix $M$ acts on the binary representation of a Pauli operator from the right. The binary representations of Pauli operators and Clifford unitaries omit the overall phases of full operators. If needed, such overall phases can always be compensated by a single layer of gates consisting solely of $Z$ and $X$ gates~\footnote{Such extra layer has depth $O(1)$. Throughout the paper, Pauli gates are assumed to be free and can be directly applied to qubits. This is also true in FTQC using stabilizer codes, where logical Pauli operators are easy to realize. } on some subsets of qubits~\cite{aaronson2004improved,maslov2017Bruhat}.

Let $\text{C}(j,l)$  denote a CNOT gate with control qubit $j$ and target qubit $l$.
The actions of appending a Hadamard, Phase, or CNOT gate to  a Clifford circuit $M$ can be described as follows:
\begin{enumerate}
  \item A Hadamard gate on qubit $j$ exchanges columns $j$ and $n+j$ of $M$.
  \item A Phase gate on qubit $j$ adds column $j$ to column $n+j$ (mod 2) of $M$.
  \item $\text{C}(j,l)$ adds column $j$ to column $l$ (mod 2) of $M$ and adds column $n+l$ to column $n+j$ (mod 2) of $M$.
\end{enumerate}
\section{Constant depth FT Clifford circuit}\label{sec:constant_depth}


\subsection{FT syndrome measurement}\label{sec:ft_syndrome}

The goal of an error correction protocol in FTQC is to find the most likely errors during computation, based on the extracted syndromes. However, the circuits to perform a syndrome measurement may introduce additional errors to the system or get wrong syndromes with high probability. Therefore, the error correction may fail, if not treated properly.

In this section, we briefly review two major protocols used in this paper --- Knill and Steane syndrome measurements~\cite{KnillFTNature, steane1997active}. Each scheme has its own advantages in different computation scenarios~\cite{chamberland2018deep}, such as a better threshold or a better ability to handle particular types of noise, and both can be used to construct arbitrary FT logical Clifford circuits.

\subsubsection{Knill syndrome measurement}

For an arbitrary $\llb n, k,d \rrb$ stabilizer code, one can use the logical teleportation circuit in Fig.~\ref{fig:Knill} to extract the error syndrome~\cite{steane1997active}, as proposed by Knill~\cite{KnillFTNature}. Here, two blocks of ancilla qubits are maximally entangled in a logical Bell state $|\Phi_L^+\>^{\otimes k}=\frac{1}{\sqrt{2}}\left(|0_L\>\otimes |0_L\>+|1_L\>
\otimes |1_L\>\right)^{\otimes k}$. The upper block of ancilla qubits are encoded to the same code protecting the data state, while the lower ones can be protected by an arbitrary stabilizer code encoding $k$ logical qubits. In this paper, we restrict ourselves to the same $\llb n, k,d \rrb$ CSS code for all blocks.

\begin{figure}[!htp]
\centering\includegraphics[width=60mm]{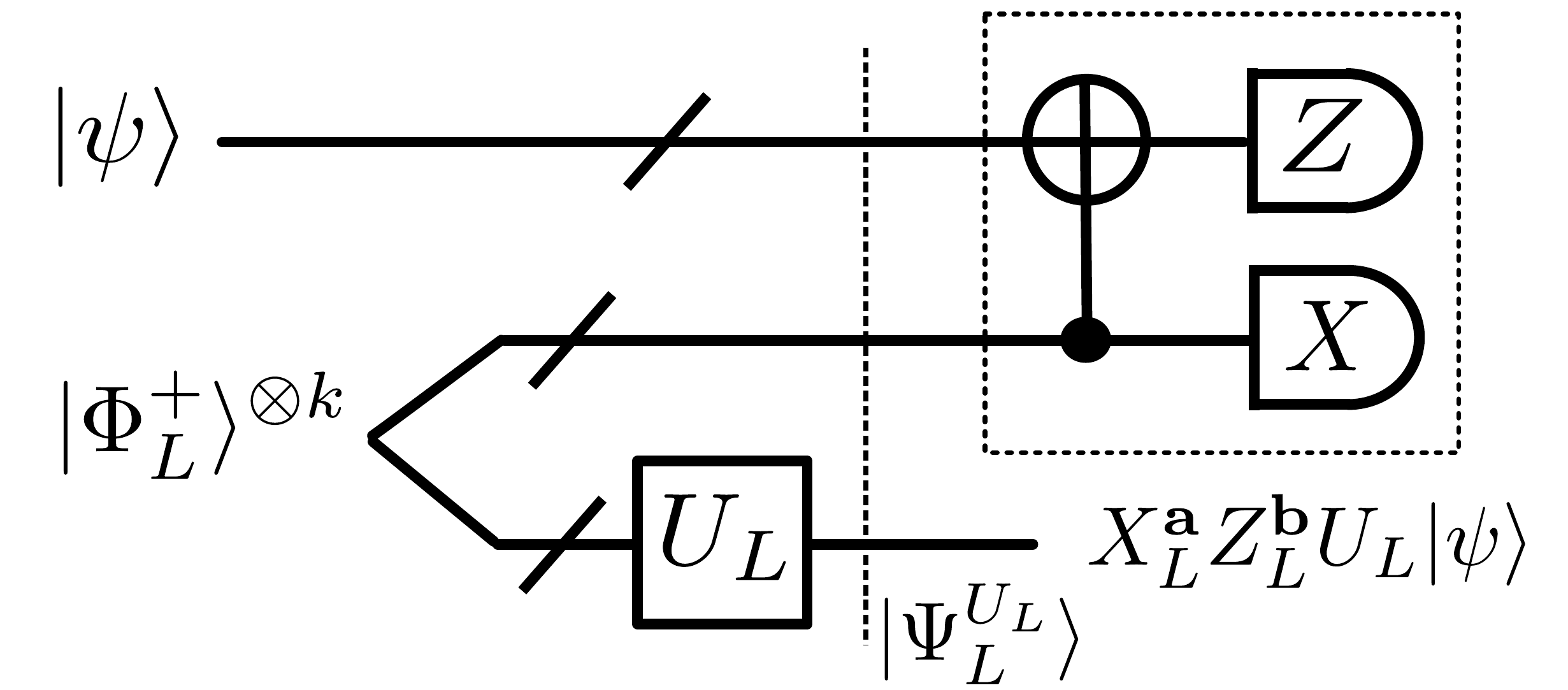}
\caption{\label{fig:Knill} The quantum circuit for Knill syndrome measurement and gate teleportation for an $\llb n, k,d \rrb$ stabilizer code. For a logical Clifford circuit $U_L$, if $|\Psi_L^{U_L}\>$ is prepared before logical Bell measurement, $U_L|\psi\>$ can be obtained up to some Pauli correction of $X_L^{\bf a}Z_L^{\bf b}$ on the output block, depending on the logical Bell measurement results.
}
\end{figure}

The logical Bell measurement in the dashed box teleports the encoded state to the lower ancilla block up to a logical Pauli correction (depending on the Bell measurement outcomes), and simultaneously obtains the error syndrome of on the input data blocks. Both logical Bell measurement outcomes and syndromes are calculated from the bitwise measurement results. The circuit is intrinsically fault-tolerant because it consists solely of transversal CNOT gates and bitwise measurements.

One particular virtue of the teleportation syndrome measurement circuit is that it can also provide a straightforward way to produce any logical circuit $U_L$ (on the teleported state) of the Clifford hierarchy $C_k$ (up to a $C_{k-1}$ correction depending on the logical measurement outcomes) via the very same syndrome measurement circuit~\cite{Gottesman:1999:390}, if one can construct the ancilla state
\beq
\left|\Psi_L^{U_L}\right\rangle=(I \otimes U_L)\left|\Phi^+ _L\right\rangle^{\otimes k}.
\eeq
This construction is very useful when implementing logical circuits for large block codes. In this paper, we focus on $U \in C_2$, the Clifford circuit. In this case, all the $|\Psi_L^{U_L}\rangle$ are stabilizer states that can be prepared by Clifford circuits.

\subsubsection{Steane syndrome measurement}

\begin{figure}[!htp]
\centering\includegraphics[width=55mm]{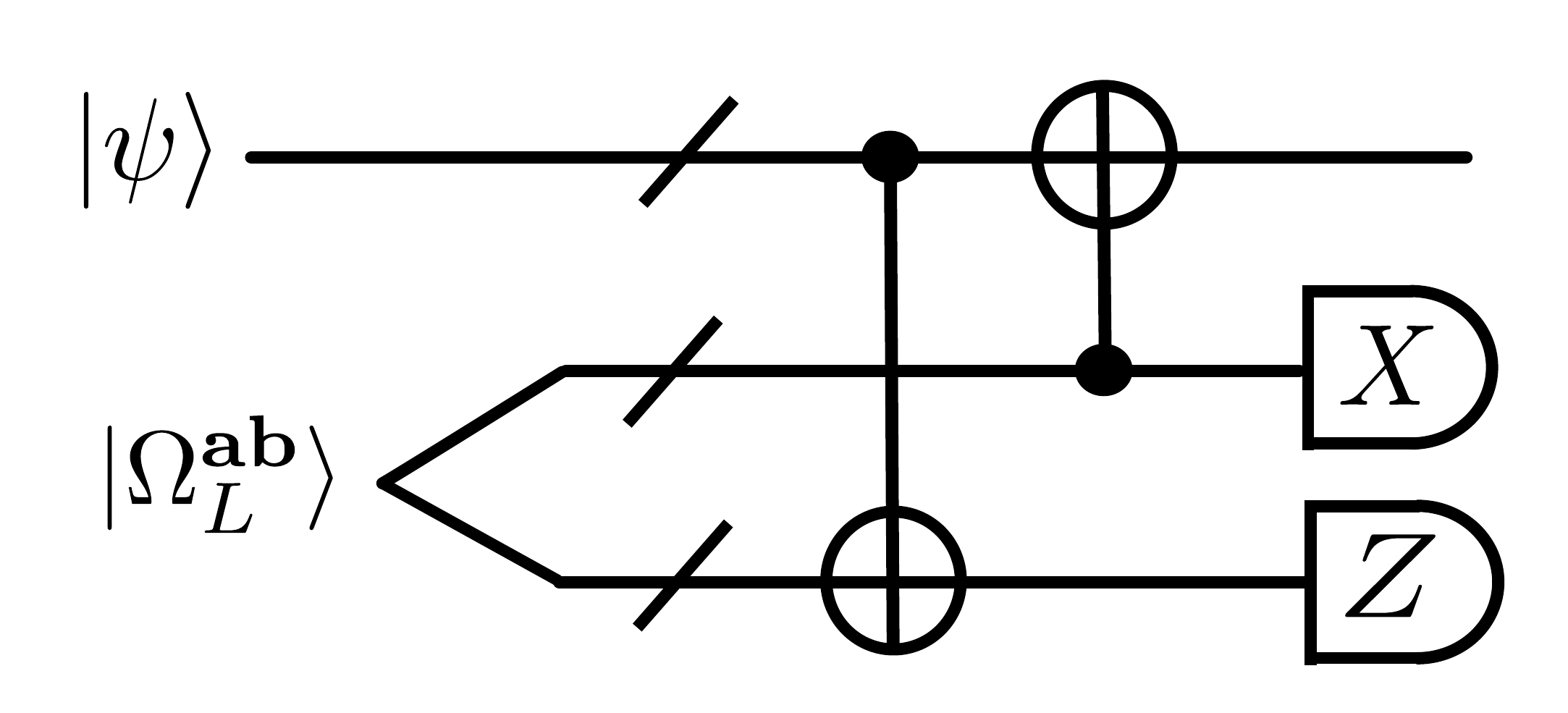}
\caption{\label{fig:Steane} The quantum circuit for Steane syndrome measurement using an $\llb n, k,d \rrb $ CSS code. One can simultaneously  measure a Hermitian Pauli operator $i^{\tau_{\bf ab}}X^{\bf a}_L Z^{\bf b}_L$, when the two ancilla blocks are prepared in the state $|\Omega^{\bf ab}_{L}\> =\frac{1}{\sqrt{2}}\left(I+i^{\tau_{\bf ab}}X^{\bf a}_L\otimes  Z^{\bf b}_L\right)|0_L\rangle^{\otimes k}\otimes |+_L\rangle^{\otimes k}$.}
\end{figure}

Now we consider a CSS code $\llb n, k,d \rrb$ for convenience in later discussion. For CSS codes, Steane suggested a syndrome measurement circuit as shown in Fig.~\ref{fig:Steane}~\cite{steane1997active}. Here, two
logical ancilla blocks of the same code are used that protects the data state. Two transversal CNOT gates propagate $Z$ and $X$ errors from the data block to ancilla blocks and corresponding error syndromes are calculated from the bitwise measurement outcomes. If the two ancilla blocks are prepared in a tensor product state $|0_L\>^{\otimes k}\otimes |+_L\>^{\otimes k}$, the circuit extracts the error syndromes without disturbing the encoded quantum information. Like the Knill syndrome measurement, the circuit is intrinsically fault-tolerant.

Moreover, one can simultaneously measure an arbitrary Hermitian logical  Pauli operator of the form $i^{\tau_{\bf ab}}X^{\bf a}_L Z^{\bf b}_L$ while extracting syndromes, if $|\Omega^{{\bf ab}}_L\rangle$ is prepared in
\beq
|\Omega^{\bf ab}_{L}\> =\frac{1}{\sqrt{2}}\left(I+i^{\tau_{\bf ab}}X^{\bf a}_L\otimes  Z^{\bf b}_L\right)|0_L\rangle^{\otimes k}\otimes |+_L\rangle^{\otimes k}.
\eeq

It is easy to prove the functionality of the circuit: start with the joint state $|\psi\>|\Omega^{{\bf ab}}_{L}\>$, after two transversal CNOTs, the state becomes
{\small
$$
\frac{1}{\sqrt{2}}\left(|\psi\> |0_L\>^{\otimes k}|+_L\>^{\otimes k} + i^{\tau_{\bf ab}} X^{\bf a}_LZ^{\bf b}_L |\psi\>  X^{\bf a}_L |0_L\>^{\otimes k} Z^{\bf b}_L|+_L\>^{\otimes k}\right).
$$
}
Let the measurement outcomes of the $j$th logical qubit in the upper and lower blocks be $v^{x}_j$ and $v^{z}_j\in \{0,1\}$, respectively. Then the joint output state is:
\begin{widetext}
\small
\beq
\begin{split}
&\frac{1}{\sqrt{2}}|\psi\>\bigotimes_{j=1}^k\left(\frac{I+(-1)^{v_j^x}X_L}{2}|0_L\>\right)
\bigotimes_{j=1}^{k}\left(\frac{I+(-1)^{v_j^z}Z_L}{2}|+_L\>\right)+\frac{1}{\sqrt{2}}i^{\tau_{\bf ab}}X_L^{\bf a}Z_L^{\bf b}|\psi\>\bigotimes_{j=1}^{k}\left(\frac{I+(-1)^{v_j^x}X_L}{2}X_L^{a_j}|0_L\>\right)
\bigotimes_{j=1}^{k}\left(\frac{I+(-1)^{v_j^z}Z_L}{2}Z_L^{b_j}|+_L\>\right)\\
=&\frac{1}{\sqrt{2}}\left(I+\prod_{l\in \text{supp}({\bf a})}(-1)^{v_l^x}\prod_{l\in \text{supp}({\bf b})}(-1)^{v_l^z}i^{\tau_{\bf ab}}X^a_LZ^b_L\right)|\psi\>\bigotimes_{j=1}^{k}\left(\frac{I+(-1)^{v_j^x}X_L}{2}|0_L\>\right)
\bigotimes_{j=1}^{k}\left(\frac{I+(-1)^{v_j^z}Z_L}{2}|+_L\>\right),\\
\end{split}
\eeq
\end{widetext}
which is the state after the measurement of $i^{\tau_{\bf ab}}X_L^{\bf a}Z_L^{\bf b}$ on $|\psi\>$ with measurement outcome
$$
\prod_{l\in \text{supp}({\bf a})}(-1)^{v_l^x}\prod_{l\in \text{supp}({\bf b})}(-1)^{v_l^z}.
$$
This circuit also allows measuring several commuting logical Pauli operator simultaneously. Here, we restrict ourselves to a commuting set of $m \leq k$ logical Pauli operators and suppose that the set of commuting Pauli operators to be measured is
$\{X_L^{{\bf e}_1}Z_L^{{\bf f}_1},\dots, X_L^{{\bf e}_m}Z_L^{{\bf f}_m} \}$.
These operators can be simultaneously measured by replacing $|\Omega^{\bf ab}_L\>$ with:
\beq\label{eq:state_multi_measurement}
|\Omega^{{\bf EF}}_L\>=\frac{1}{\sqrt{2^m}}\prod_{j=1}^{m}
\left(I+i^{\tau_{{\bf e}_j{\bf f}_j}}X_L^{{\bf e}_j}\otimes Z_L^{{\bf f}_j}\right)|0_L\>^{\otimes k}\otimes |+_L\>^{\otimes k}.
\eeq
Note that $|\Omega^{{\bf EF}}_L\>$ is also a stabilizer state. Like logical circuit teleportation, one can also effectively construct any logical Clifford circuit via such Pauli measurements~\cite{Gottesman:9705052,brun2015teleportation}.

\subsection{Single-shot FT logical circuit teleportation and Pauli measurement}
Ideally, if the ancilla qubits are clean and measurements are perfect, one can extract the error syndrome of the data block with logical circuit teleportation or Pauli measurements in a single round of Knill/Steane syndrome measurement.

In practice, ancillas may contain different types of errors after preparation, while the measurement outcomes can also be noisy. One needs to make sure that high weight errors do not propagate from ancilla qubits to data blocks. At the same time, reliable values of syndromes and logical operators must be established from measurement outcomes. For error correction, one can repeat the syndrome measurements several rounds to establish reliable syndromes of the data state via majority vote~\cite{Shor:1996:56, Steane:2003:042322}. However, for the purpose of logical circuit teleportation or Pauli measurements, one needs reliable values of logical operators right after the first round of measurement for subsequent correction. Otherwise, it will cause a logical error on the data state. Thus, a single-shot FT logical circuit teleportation or Pauli measurement protocol is required.

Actually, we will see this is possible if the blocks of ancilla qubits for Knill/Steane syndrome measurements do not contain any \emph{correlated} errors. Here, we define \red{an} uncorrelated error as follow~\cite{steane2002fast}:
\begin{mydef}\label{def:uncorrelatation}
For an $\llb n,k,d \rrb$ code correcting any Pauli error on $t=\lfloor \frac{d-1}{2} \rfloor$ qubits, we say that an error $E$ on the code block is spatially \emph{uncorrelated} if the probability of $E$ is
$$
\Pr(E)\sim O(p^s):
\begin{cases}
&\text{for some }s\geq\text{wt}(E),\ \text{if wt}(E)\leq t; \\
&\text{for some }s \geq t,\ \text{if wt}(E)> t,
\end{cases}
$$
where the coefficients behind $O$ are not unreasonably large.
\end{mydef}
Otherwise, $E$ is said to be correlated. For those uncorrelated errors satisfying this definition, they should have a distribution similar to the binomial distribution. Thus, the errors on the code block can be regarded as independent. We say that an ancilla is \emph{qualified} if it is free of correlated errors.

\begin{figure}[!htp]
\centering\includegraphics[width=60mm]{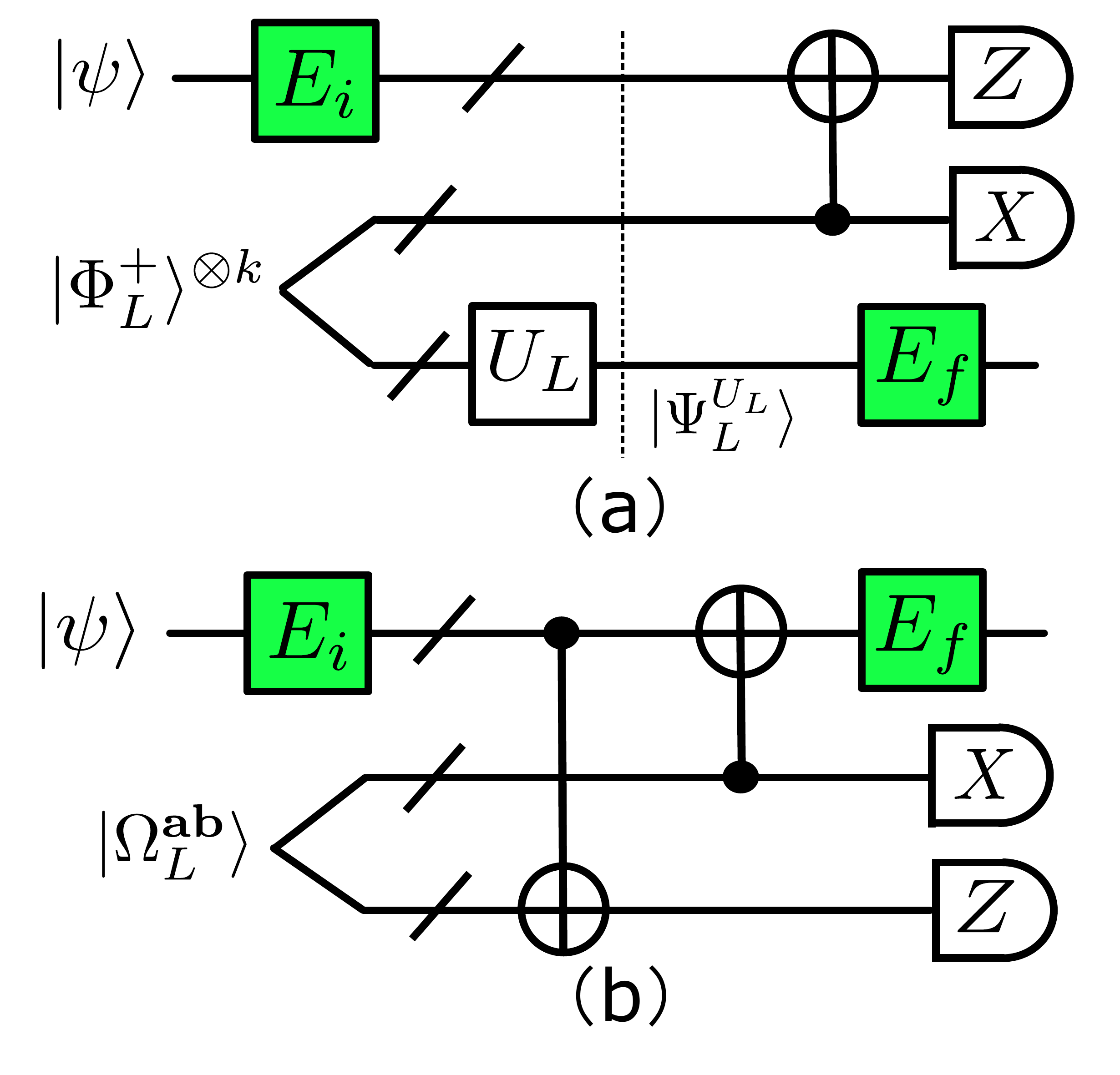}
\caption{\label{fig:equivalent_error} The effective error model of the Knill (part (a)) and Steane (part(b)) syndrome measurement circuits. While $E_i$ and $E_f$ are in general correlated \emph{in time}, they are \emph{spatially} uncorrelated, if the ancilla states are qualified.
}
\end{figure}
It is obvious that no correlated error will be propagated back to the data blocks if ancilla blocks are qualified. Even more, we have:
\begin{mylemma}[Effective error model]\label{lemma:effective_error}
During imperfect logical state teleporation via Knill syndrome extraction, or logical Pauli measurements via Steane syndrome measurement, if errors in the same block (data or ancilla) are spatially uncorrelated according to Def.~\ref{def:uncorrelatation}, then the errors are equivalent to spatially \emph{uncorrelated} effective errors acting only on the data code block  before and after the process, as shown in Fig.~\ref{fig:equivalent_error}. 
\end{mylemma}

It has already been shown in Ref.~\cite{brun2015teleportation} that this statement is true for Steane syndrome measurement. The basic idea is that failures occurring in any location of the circuit can be commuted forward or backward to the data code block, allowing the ancillas to be treated as clean and the measurements as perfect. Thus we can leave $E_f$ to the next round of syndrome measurements and analyze as if only $E_i$ (and $E_f$ from the previous round) have occurred, followed by perfect syndrome measurements. The same argument is also applicable to Knill syndrome measurements and hence one has:
\begin{mytheorem}\label{thm:single-shot}
The Knill/Steane syndrome measurement circuit can implement fault-tolerant logical circuit teleportation/Pauli measurements in a single round.
\end{mytheorem}
After a single-shot Knill/Steane syndrome measurement and correction, the final data state is:
\beq\label{eq:final_state}
|\psi_f\> \propto E_f \cdot  R \cdot O_L \cdot L_C \cdot E_{\text{comb}}|\psi_i\>.
\eeq
Here, $E_{\text{comb}}$ includes both $E_i$ in current stage and $E_f$ from previous stage; $O_L$ is the logical operation (either a teleported logical circuit, or logical Pauli measurements); $L_C$ is the Pauli correction based on the outcomes of logical measurements; $R$ is the recovery operation based on the measured syndromes, $O_L$ and $L_C$.


%


\subsection{Constant depth Clifford circuit via FT circuit teleportation}\label{sec:constant_teleport}
For a CSS code with $k$ logical data qubits, it requires $O(k^2/\log k)$ logical Clifford gates~\cite{aaronson2004improved, markov2008optimal} for a logical Clifford circuit. If we implement these gates one by one in a fault-tolerant manner, it will require $O(k^2/\log k)$ qualified ancilla states using $O(k^2/\log k)$ times of the Knill/Steane single-shot syndrome measurements circuit. In this and next subsections, we show that $O(1)$ qualified ancilla states and $O(1)$ steps of the Knill/Steane syndrome measurements are sufficient for arbitrary logical Clifford circuits, up to a permutation of qubits.

It is well known that any Clifford circuit has an equivalent circuit comprising 11 stages, each using only one type of gate: -H-C-P-C-P-C-H-P-C-P-C-~\cite{aaronson2004improved}. That can be further reduced to a 9-stage -C-P-C-P-H-P-C-P-C-~\cite{maslov2017Bruhat}. More specifically, one has:
\begin{mytheorem}[Bruhat decomposition~\cite{maslov2017Bruhat}]\label{thm:bruhat}
Any symplectic matrix $M$ of dimension $2k\times 2k$ can be decomposed as
\beq\label{eq:used_bruhat}
\begin{split}
M=&M^{(1)}_{C} M^{(1)}_{P}M^{(2)}_{C}M^{(2)}_{P} M^{(1)}_{H} \cdot\\
&M^{(3)}_{P} \left(\pi M^{(3)}_{C}\pi^{-1}\right)M^{(4)}_{P} \left(\pi M^{(4)}_{C}\pi^{-1}\right)\pi.
\end{split}
\eeq
Here, $M^{(j)}_{C} $ are -C- stage matrices containing only CNOT gates $\text{C}(q,r)$ such that $q < r$; $M^{(j)}_{P} $ and $M^{(j)}_{H} $ are -P- and -H- stage matrices; $\pi$ is a permutation matrix.
\end{mytheorem}

In a -P- stage, since $\text{P}^4=I_2$, effectively there are three types of single-qubit gates: P, P$^2=Z$ and $\text{P}^3=\text{P}^\dag=\text{P}Z$. Note that we will postpone all the $Z$ gates to the final stage, and thus the -P- layer consists of at most $k$ individual Phase gates. The symplectic matrix of a -P- stage on a set of $m$ qubits is in general of the form:
\beq
M_P=\left(
  \begin{array}{c|c}
    I_k      &  \Lambda_k^m \\[2pt]
    \hline
    \rule[0.4ex]{0pt}{8pt}
    {\bf 0}_k &  I_k\\
  \end{array}
\right),
\eeq
where $\Lambda_k^m$ is an $k\times k$ diagonal matrix with $m$ 1s.

Similar to the -P- stage,  since $\text{H}^2=I_2$, an -H- stage  contains at most $k$ individual H gates. The symplectic matrix of an -H- stage on an arbitrary set of $m$ qubits can be written as
\beq
M_H=\left(
  \begin{array}{c|c}
    I_k + \Lambda_k^m & \Lambda_k^m \\[2pt]
    \hline
    \rule[1.0ex]{0pt}{8pt}
    \Lambda_k^m & I_k+\Lambda_k^m \\
  \end{array}
\right).
\eeq

The corresponding symplectic matrix of a -C- stage can be written as:
\beq
M_C=\left(
  \begin{array}{c|c}
     U      &  {\bf 0}_k\\[2pt]
    \hline
    \rule[1.0ex]{0pt}{8pt}
    {\bf 0}_k &  \left( U^{t}\right)^{-1}\\
  \end{array}
\right),
\eeq
where $U$ is an invertible $k\times k$ upper triangular matrix.

Clearly, if one can implement each stage in $O(1)$ steps fault-tolerantly, an arbitrary logical Clifford circuit can be implemented in $O(1)$ steps. For Knill syndrome measurements, it is straightforward---one can prepare the ancilla for the circuit in each stage directly as:
\beq
\begin{split}
\left|\Psi_{L}^{U_P}\right\>&=I\otimes U_P(|0_L\>\otimes|0_L\>+|1_L\>\otimes |1_L\>)^{\otimes k},\\
\left|\Psi_{L}^{U_H}\right\>&=I\otimes U_H(|0_L\>\otimes|0_L\>+|1_L\>\times |1_L\>)^{\otimes k},\\
\left|\Psi_{L}^{U_C}\right\>&=I\otimes U_C(|0_L\>\otimes|0_L\>+|1_L\>\otimes|1_L\>)^{\otimes k}
\end{split}
\eeq
where $U_P$, $U_H$ and $U_C$ are the corresponding unitaries for the -P-, -H- and -C- stages, respectively. Obviously, these are all CSS states up to local Clifford operations, whose binary representations at the logical level are:
\beq
\Psi_L^{U_P}=\left(
  \begin{array}{cc|cc}
    I_k & I_k & {\bf 0} & \Lambda_k^m \\[2pt]
    {\bf 0} & {\bf 0} & I_k & I_k\\
  \end{array}
\right),
\eeq


\beq
\Psi_L^{U_H}=\left(
                   \begin{array}{cc|cc}
                     I_k & I_k+\Lambda_k^m & {\bf 0} & \Lambda_k^m \\[2pt]	                 {\bf 0} & \Lambda_k^m & I_k & I_k+\Lambda_k^m \\
                   \end{array}
                 \right)
\eeq
assuming -P- or -H- is applied to a set of $m$ qubits, and
\beq
\Psi_L^{U_C}=\left(
  \begin{array}{cc|cc}
    I_k & U & {\bf 0} & {\bf 0} \\
    {\bf 0} & {\bf 0} & I_k &  \left(U^t\right)^{-1} \\
  \end{array}
\right).
\eeq
If these states are all qualified for all the stages, by Theorem.~\ref{thm:single-shot} and \ref{thm:bruhat}, one can implement an arbitrary logical Clifford circuit in 9 rounds of single-shot Knill syndrome measurements. Later, we will show that all the three types of ancilla states can be prepared fault-tolerantly and efficiently.

\subsection{Constant depth Clifford circuit FT Pauli measurement}\label{sec:constant_pauli_measurement}
Unlike Knill syndrome measurement, it is not so obvious how to implement the logical Clifford group using Steane syndrome measurement. In this section, we provide a constructive proof showing that by introducing $k$ extra auxiliary logical qubits (labeled as $\text{A}_1,\dots, \text{A}_k$), each stage of a logical Clifford circuit on $k$ data logical qubits ($\text{Q}_1,\dots, \text{Q}_k$) can be implemented via a constant number of Pauli operator measurements, up to a permutation of qubits. We choose an $\llb n, 2k, d\rrb$ CSS code and put the logical qubits in the following order:  $\{\text{A}_1,\dots, \text{A}_k, \text{Q}_1,\dots, \text{Q}_k\}$.

\subsubsection{-P- stage}
Consider a pair of qubits $\{\text{A}_j,\text{Q}_j\}$ with $\text{A}_j$ in the $|0_L\>$ state. Measure operators $X_{{\text{A}_j},L}Y_{\text{Q}_{j},L}$ and then $Z_{\text{Q}_j,L}$. After swapping $\text{A}_j$ and $\text{Q}_j$, the overall effect is a Phase gate
on $\text{Q}_j$ up to a Pauli correction depending on the measurement
outcomes. The swap does not need to be done physically.
Instead, one can just keep a record of it in software.

For $m$ Phase gates on a logical qubit set $\mathscr{M}$, since $\{X_{{\text{A}_j},L}Y_{\text{Q}_j,L}|\ j\in\mathscr{M}\}$ and $\{Z_{\text{Q}_j,L} |\ j\in\mathscr{M}\}$ are commuting operator sets, it requires only two steps of Pauli measurements by preparing two ancilla states with $4k$ logical qubits:
{\small
\beq\label{eq:phase_ancilla}
\begin{split}
&\left|\Omega^{P_1}_L\right\>\\
=&\frac{1}{\sqrt{2^m}}\prod_{j\in \mathscr{M}}
\left(I+i\left(X_{j,L}X_{j+k,L}\right)\otimes Z_{j+k,L}\right)|0_L\>^{\otimes 2k}\otimes |+_L\>^{\otimes 2k}
\end{split}
\eeq
}
and
\beq
\left|\Omega^{P_2}_L\right\>=\frac{1}{\sqrt{2^m}}|0_L\>^{\otimes 2k}\otimes \left(\prod_{j\in \mathscr{M}}\left(I+ Z_{j,L}\right)|+_L\>^{\otimes 2k}\right),
\eeq
whose binary representations at the logical level are
\beq
\Omega^{P_1}_L=\left(
  \begin{array}{cccc|cccc}
    {\bf 0} & {\bf 0} & {\bf 0} & {\bf 0} & I_k & \Lambda_k^m & {\bf 0} & {\bf 0} \\[2pt]
    \Lambda_k^m &\Lambda_k^m &{\bf 0} &{\bf 0} &{\bf 0} & I_k &{\bf 0} & {\Lambda}_k^m \\[2pt]
    {\bf 0} & {\bf 0} & I_k & {\bf 0} & {\bf 0}& {\bf 0} & {\bf 0}_k & {\bf 0}\\[2pt]
    {\bf 0}& {\bf 0} & {\bf 0} & I_k & {\bf 0} & \Lambda_k^m & {\bf 0}&  {\bf 0}_k\\
  \end{array}
\right),
\eeq
and
\beq
\Omega^{P_2}_L=\left(
\begin{array}{ccc|ccc}
{\bf 0}_{2k} & {\bf 0} & {\bf 0} & I_{2k}  & {\bf 0} & {\bf 0} \\[2pt]
{\bf 0} & I_k & {\bf 0 } & {\bf 0} & {\bf 0}_k& {\bf 0} \\[2pt]
{\bf 0}& {\bf 0} & I_k+\Lambda_k^m & {\bf 0}& {\bf 0} & \Lambda_{k}^m  \\
\end{array}
\right),
\eeq
respectively. Note that $|\Omega_L^{P_2}\>$ is a CSS state. $|\Omega_L^{P_1}\>$ is the joint $+1$ eigenstate of
{\small
$$
\{Z_{j,L} Z_{j+k,L}\otimes I_{2n}, Z_{j+k,L}\otimes X_{j+k,L},X_{j,L}Y_{j+k,L}\otimes Z_{j+k,L}\ |\ j\in \mathscr{M}\},
$$
}
which is also a CSS state up to Phase gates on logical qubits $\{j+k| \ j\in \mathscr{M}\}$ of the upper block, and Hadamard gates on the logical qubits  $\{j+k| \ j\in \mathscr{M}\}$ of the lower block.

\subsubsection{-H- stage}
Like the -P- stage, we  consider only a single H on a data qubit. For a pair of qubits $\{\text{A}_j,\text{Q}_j\}$ with $\text{A}_j$ in the $|0_L\>$ state,
measure $X_{\text{A}_j,L}Z_{\text{Q}_j,L}$ and then $X_{\text{Q}_j,L}$.
After swapping $\text{A}_j$ and $\text{Q}_j$ , the overall effect is a Hadamard
gate on $\text{Q}_j$ with $\text{A}_j$ in the  $|+_L\>$ up to a Pauli correction depending
on the measurement outcome.

For $m$ Hadamard gates on a logical qubits set $\mathscr{M}$,
since $\{X_{\text{A}_j,L}Z_{\text{Q}_j,L}\ | \ j\in \mathscr{M} \}$ and $\{X_{\text{Q}_j,L}\ | \ j\in \mathscr{M}\}$ are both commuting sets, we need just two steps of Pauli measurements and an ancilla state with $4k$ logical qubits for an -H- stage. If Hadamard gates are applied to a set $\mathscr{M}$ of qubits, the required ancilla states are
\beq\label{eq:had_ancilla}
\left|\Omega^{H_1}_L\right\>=\frac{1}{\sqrt{2^m}}\prod_{j\in \mathscr{M}}
\left(I+X_{j,L}\otimes Z_{j+k,L}\right)|0_L\>^{\otimes 2k}\otimes |+_L\>^{\otimes 2k},
\eeq
and
\beq
\left|\Omega^{H_2}_L\right\>=\frac{1}{\sqrt{2^m}}\left(\prod_{j\in \mathscr{M}}
\left(I+ X_{j,L}\right)|0_L\>^{\otimes 2k}\right)\otimes |+_L\>^{\otimes 2k},
\eeq
whose binary representations at the logical level are:
\beq
\Omega_L^{H_1}=\left(
  \begin{array}{cccc|cccc}
    \Lambda_k^m & {\bf 0} & {\bf 0} & {\bf 0} & I_k+\Lambda_k^m & {\bf 0} & {\bf 0} & \Lambda_k^m\\[2pt]
    {\bf 0} & {\bf 0}_k & {\bf 0} & {\bf 0} & {\bf 0} & I_k & {\bf 0} & {\bf 0} \\[2pt]
    {\bf 0} & {\bf 0} & I_k & {\bf 0} & {\bf 0} & {\bf 0} & {\bf 0}_k & {\bf 0} \\[2pt]
    {\bf 0} & {\bf 0} & {\bf 0} & I_k & \Lambda_k^m & {\bf 0} & {\bf 0} & {\bf 0} \\
  \end{array}
\right)
\eeq
and
\beq
\Omega^{H_2}_L=\left(
\begin{array}{ccc|ccc}
{\bf 0}_k & {\bf 0 } & {\bf 0} & I_k& {\bf 0} & {\bf 0} \\[2pt]
{\bf 0} & \Lambda_k^m & {\bf 0} & {\bf 0} & I_{k}+\Lambda_k^m& {\bf 0}\\[2pt]
{\bf 0}& {\bf 0}&I_{2k}& {\bf 0}& {\bf 0} & {\bf 0}_{2k}
\end{array}
\right),
\eeq
respectively. Note that $\left|\Omega^{H_2}_L\right\>$ is a CSS state. $\left|\Omega_L^{H_1}\right\>$ is the joint $+1$ eigenstate of $\{X_{j,L}\otimes Z_{j+k,L}, Z_{j,L}\otimes X_{j+k,L}\ |\ j\in \mathscr{M}\}$, and thus, it is a CSS state up to Hadamard gates.


\subsubsection{-C- stage}
We first introduce the generalized stabilizer formalism that is helpful later. Consider a $2^k$ dimensional subspace $C(\mathcal{G})$ of the $N$ logical qubit Hilbert space, where $\cG$ has $N-k$ stabilizer generators. We focus on the effects of Clifford circuits on the $k$ logical qubits stabilized by $\mathcal{G}$. Consider a set of matrices $\textsf{C}_{\mathcal{G}}$ of the form:
\beq\label{eq:generalized_circuit matrix}
\left(
  \begin{array}{c|c}
    Q' & R' \\
    \hline
    \rule[0.4ex]{0pt}{8pt}
    S' & T'  \\
    \hline
    \rule[0.4ex]{0pt}{8pt}
    A & B
  \end{array}
\right).
\eeq
Here, $(A|B)$ corresponds to the stabilizer generators of $\mathcal{G}$;   $(Q'|R')$ and $(S'|T')$ are $k\times 2N$ binary matrices orthogonal to $(A|B)$ with respect to the symplectic inner product, and which are symplectic partners of each other. They can  be regarded as ``encoded operators" on $C(\mathcal{G})$.
We define the following equivalence relation $\mathscr{R}$ in  $\textsf{C}_\mathcal{G}$: two matrices
\beqs
{C}_1=\left(
  \begin{array}{c|c}
    Q'_1& R'_1 \\[2pt]
    \hline
    \rule[0.4ex]{0pt}{8pt}
    S'_1 & T'_1  \\[2pt]
    \hline
    \rule[0.4ex]{0pt}{8pt}
    A_1 & B_1
  \end{array}
\right) \quad \text{and} \quad {C}_2=\left(
  \begin{array}{c|c}
    Q'_2& R'_2 \\[2pt]
    \hline
    \rule[0.4ex]{0pt}{8pt}
    S'_2 & T'_2  \\[2pt]
    \hline
    \rule[0.4ex]{0pt}{8pt}
    A_2 & B_2
  \end{array}
\right),
\eeqs
are equivalent if (a)~$(A_1|B_1)$ and $(A_2|B_2)$ generate the same stabilizer group $\mathcal{G}$; and (b)~{\tiny
$\left(
  \begin{array}{c|c}
    Q'_1 & R'_1 \\[2pt]
    \hline
    \rule[0.4ex]{0pt}{5pt}
    S'_1 & T'_1  \\
  \end{array}
\right)$} differs from {\tiny
$\left(
  \begin{array}{c|c}
    Q'_2 & R'_2 \\[2pt]
    \hline
    \rule[0.4ex]{0pt}{5pt}
    S'_2 & T'_2  \\
  \end{array}
\right)$} by multiplication of elements in $\mathcal{G}$. Thus, there is a one-to-one correspondence between $\textsf{C}_{\mathcal{G}}/\mathscr{R}$ and ${\rm Sp}(2k, \mathbb{Z}_2)$.
Therefore, $\textsf{C}_{\mathcal{G}}/\mathscr{R}$ captures the behavior of stabilizer circuits on $C(\mathcal{G})$. The circuit representation of Eq.~(\ref{eq:generalized_circuit matrix}) is called the \emph{generalized stabilizer form} (GSF) of $\mathcal{G}$.

The following lemma will also be used in the circuit construction:
\begin{mylemma}\label{lemma:stabilizer_transform}
Let $L_1$ be an $n\times n$ lower triangular matrix with the diagonal elements being zeros. 
Suppose
\beqs
L=(I_n  \ L_1).
\eeqs
Then there exists  a full-rank matrix $L'=(L_2  \ L_3)$, where $L_2$ and $L_3$ are two $n\times n$ lower triangular matrices,
such that the rows of $L'$ are linear combinations of rows of $L$ and
\beq\label{eq:orthogonal}
 L'
 \left(
   \begin{array}{c}
     I_n \\
     I_n \\
   \end{array}
 \right) = L_2+L_3= I_n.
 \eeq
\end{mylemma}
\begin{proof}
See Appendix.~\ref{sec:proof_lemma} for details.
\end{proof}

We now construct the sequence of Pauli measurements which can generate any -C- stage on logical qubits $\text{Q}_1,\dots , \text{Q}_k$ of an $\llb n,N=2k,d\rrb$ CSS code. We
start with the GSF of an arbitrary -C- circuit with auxiliary logical qubits $\text{A}_1,\dots,\text{A}_k$ in $|+_L\>^{\otimes k}$:
\beq
\left(
  \begin{array}{cc|cc}
    {\bf 0}_k  & U & {\bf 0}_k  & {\bf 0}_k \\[2pt]
    \hline
    \rule[0.4ex]{0pt}{8pt}
    {\bf 0}_k  & {\bf 0}_k & {\bf 0}_k & (U^t)^{-1} \\[2pt]
    \hline
    \rule[0.4ex]{0pt}{8pt}
    I_k & {\bf 0}_k & {\bf 0}_k & {\bf 0}_k \\
  \end{array}
\right),
\eeq
and reduce it to the idle circuit by a series of row operation. This set of operations in reverse will effectively implement the target CNOT circuit.

As mentioned before, $U$ is an invertible upper triangular matrix. The GSF is then equivalent to
\beq\label{eq:first_measure_before}
\left(
  \begin{array}{cc|cc}
    U+I_k  & U & {\bf 0}_k  & {\bf 0}_k \\[2pt]
    \hline
    \rule[0.4ex]{0pt}{8pt}
    {\bf 0}_k  & {\bf 0}_k & {\bf 0}_k & (U^t)^{-1} \\[2pt]
    \hline
    \rule[0.4ex]{0pt}{8pt}
    I_k & {\bf 0}_k & {\bf 0}_k & {\bf 0}_k \\
  \end{array}
\right)
\eeq
since all the nonzero row vectors of $\left(U+I_k \ \ {\bf 0}_k \ | \ {\bf 0}_k  \ \ {\bf 0}_k\right)$ can be generated by $\left(I_k \ \ {\bf 0}_k \ | \ {\bf 0}_n  \ \ {\bf 0}_n\right)$
and we add these vectors to the first row.

Since $U$ is of full rank, the diagonal elements of $U+I_k$ must be all zeros.
Observe that  $\left({\bf 0}_k \ {\bf 0}_k \ | \ {I}_k \ \ (U^t)^{-1}+I_k\right)$ commutes with the logical operators and is a symplectic partner of the stabilizer generators, since
\beqs
\left( {I}_k \ \ \ \ (U^t)^{-1}+I_k\right) \left(U+I_k \ \ \ \ U\right)^t = {\bf 0}_k,
\eeqs
and
\beqs
\left(I_k \ {\bf 0}_k \ |{\bf 0}_k \ \ {\bf 0}_k\right)\left(\ {I}_k \ \ (U^t)^{-1}+I_k \ | \ {\bf 0}_k \ {\bf 0}_k \right)^t=I_{2k}.
\eeqs
One can measure $k$ commuting logical Pauli operators $\left({\bf 0}_k \ {\bf 0}_k \ | \ {I}_k \ \ (U^t)^{-1}+I_k\right)$ simultaneously. The GSF will then be transformed into
\beq\label{eq:first_measure_after}
\left(
  \begin{array}{cc|cc}
    U+I_k  & U & {\bf 0}_k  & {\bf 0}_k \\[2pt]
    \hline
    \rule[0.4ex]{0pt}{8pt}
    {\bf 0}_k  & {\bf 0}_k & {\bf 0}_k & (U^t)^{-1} \\[2pt]
    \hline
    \rule[0.4ex]{0pt}{8pt}
    {\bf 0}_k & {\bf 0}_k & {I}_k & (U^t)^{-1}+I_k \\[2pt]
  \end{array}
\right).
\eeq
(Meanwhile, we can perform the Pauli measurements $\left(I_k \ \ {\bf 0}_k \ | \ {\bf 0}_k \ \ {\bf 0}_k\right)$ to reverse the process (from Eq.~(\ref{eq:first_measure_after}) to Eq.~(\ref{eq:first_measure_before})).)

Now, adding the third row of Eq.~(\ref{eq:first_measure_after}) to the second row, one can obtain an equivalent GSF
\beq
\left(
  \begin{array}{cc|cc}
    U+I_k  & U & {\bf 0}_k  & {\bf 0}_k \\[2pt]
    \hline
    \rule[0.4ex]{0pt}{8pt}
    {\bf 0}_k  & {\bf 0}_k & I_k & I_k \\[2pt]
    \hline
    \rule[0.4ex]{0pt}{8pt}
    {\bf 0}_k & {\bf 0}_k & {I}_k & (U^t)^{-1}+I_k \\
  \end{array}
\right).
\eeq
Let $L=\left( {I}_k \ \ \ L_1\right)$, where $L_1=(U^t)^{-1}+I_n$ is a lower triangular matrix with all the diagonal elements being~0. By Lemma~\ref{lemma:stabilizer_transform}, the GSF can be equivalently transformed into
\beq\label{eq:second_measure_before}
\left(
  \begin{array}{cc|cc}
    U+I_k  & U & {\bf 0}_k  & {\bf 0}_k \\[2pt]
    \hline
    \rule[0.4ex]{0pt}{8pt}
    {\bf 0}_k  & {\bf 0}_k & I_k & I_k \\[2pt]
    \hline
    \rule[0.4ex]{0pt}{8pt}
    {\bf 0}_k & {\bf 0}_k & L_2 & L_3 \\
  \end{array}
\right),
\eeq
where $(L_2 \ \ L_3) (I_k \ \ I_k)^t=I_k$. One can measure a set of $k$ Pauli operators $(I_k \ \ I_k | \ {\bf 0}_k \ \  {\bf 0}_k)$ simultaneously and transform the GSF into
\beq\label{eq:second_measure_after}
\left(
  \begin{array}{cc|cc}
    U+I_k  & U & {\bf 0}_k  & {\bf 0}_k \\[2pt]
    \hline
    \rule[0.4ex]{0pt}{8pt}
    {\bf 0}_k  & {\bf 0}_k & I_k & I_k \\[2pt]
    \hline
    \rule[0.4ex]{0pt}{8pt}
    I_k & I_k & {\bf 0}_k & {\bf 0}_k \\
  \end{array}
\right).
\eeq
Meanwhile, measuring $\left({\bf 0}_k \ \ {\bf 0}_k \ | \ L_2 \ \ L_3\right)$  will transfer the GSF of Eq.~(\ref{eq:second_measure_after}) into Eq.~(\ref{eq:second_measure_before}). Note that the measurement of $\left({\bf 0}_k \ \ {\bf 0}_k \ | \ L_2 \ \ L_3\right)$ is equivalent to measuring $\left({\bf 0}_k \ \ {\bf 0}_k \ | \ {I}_k \ \ (U^t)^{-1}+I_k \right)$.

Now, since the stabilizer generators in Eq.~(\ref{eq:second_measure_after})  are of the form $\left(I_k \ \ I_k \ | \ {\bf 0}_k \ \ {\bf 0}_k\right)$, one can add $\left(U+I_k \ \ U+I_k \ | \ {\bf 0}_k  \ \ {\bf 0}_k\right)$ to the first row of Eq.~(\ref{eq:second_measure_after}), which equivalently reduces the GSF to:
\beq\label{eq:third_measure_before}
\left(
  \begin{array}{cc|cc}
    {\bf 0}_k  & I_k & {\bf 0}_k  & {\bf 0}_k \\[2pt]
    \hline
    \rule[0.4ex]{0pt}{8pt}
    {\bf 0}_k  & {\bf 0}_k & I_k & I_k \\[2pt]
    \hline
    \rule[0.4ex]{0pt}{8pt}
    I_k & I_k & {\bf 0}_k & {\bf 0}_k \\
  \end{array}
\right).
\eeq

The final step is to eliminate the left-most $I_k$ in the second row of Eq.~(\ref{eq:third_measure_before}). This can be done by measuring  $\left({\bf 0}_k \ \ {\bf 0}_k \ | \ I_k  \ \ {\bf 0}_k\right)$ and adding the third row to the second. This will then transform the GSF into the second matrix in:
\beq\label{eq:third_measure_after}
\left(
  \begin{array}{cc|cc}
    {\bf 0}_k  & I_k & {\bf 0}_k  & {\bf 0}_k \\[2pt]
    \hline
    \rule[0.4ex]{0pt}{8pt}
    {\bf 0}_k  & {\bf 0}_k & {\bf 0}_k & I_k \\[2pt]
    \hline
    \rule[0.4ex]{0pt}{8pt}
    {\bf 0}_k & {\bf 0}_k & I_k & {\bf 0}_k \\
  \end{array}
\right),
\eeq
Meanwhile, one can measure the set of $k$ logical Pauli operators $\left(I_k \ \ I_k \ | \ {\bf 0}_k \ \ {\bf 0}_k\right)$ to transform Eq.~(\ref{eq:third_measure_after}) to Eq.~(\ref{eq:third_measure_before}).

To reverse the whole procedure above and start from Eq.~(\ref{eq:third_measure_after}), we initially set $\text{A}_1,\dots, \text{A}_k$ to $|0_L\>^{\otimes k}$ and perform the following three sets of Pauli measurements:
\beq\label{eq:three_measurements}
\begin{split}
&1.~\left(I_k \ \ I_k \ | \ {\bf 0}_k \ \ {\bf 0}_k\right),\\
&2.~\left({\bf 0}_k \ \ {\bf 0}_k \ | \ {I}_k \ \ (U^t)^{-1}+I_k \right),\\
&3.~\left(I_k \ \ {\bf 0}_k \ | \ {\bf 0}_k \ \ {\bf 0}_k\right).
\end{split}
\eeq
The measurements require three $4k$ logical qubits CSS ancilla states, which are
\beq\label{eq:Cnot_ancilla_1}
\left|\Omega_L^{C_1}\right\>=\frac{1}{\sqrt{2^k}}\left(\prod_{j=1}^k \left(I+X_{j,L} X_{j+k,L}\right)|0_L\>^{\otimes 2k}\right)\otimes |+_L\>^{\otimes 2k},
\eeq
\beq\label{eq:Cnot_ancilla_2}
\left|\Omega_L^{C_2}\right\>=\frac{1}{\sqrt{2^k}} |0_L\>^{2k}\otimes\left(\prod_{j=1}^k \left(I+Z^{{\bf u}_j}_L\right)|+_L\>^{\otimes 2k})\right),
\eeq
and
\beq
\left|\Omega_L^{C_3}\right\> = \left(|+_L\>^{\otimes k}|0_L\>^{\otimes k}\right)\otimes |+_L\>^{\otimes 2k}.
\eeq
Here, ${\bf u}_j$ is the $j$th row of   $\left({I}_k \quad (U^t)^{-1}+I_k \right)$. The binary representations of these states are:
\beq
\Omega^{C_1}_L= \left(
            \begin{array}{ccc|ccc}
              I_k & I_k & {\bf 0} & {\bf 0} & {\bf 0} & {\bf 0}\\[2pt]
              {\bf 0} & {\bf 0} & {\bf 0} & I_k & I_k& {\bf 0} \\[2pt]
              {\bf 0}& {\bf 0}& I_{2k} & {\bf 0} & {\bf 0}& {\bf 0}_{2k}\\
            \end{array}
          \right),
\eeq

\beq
\Omega^{C_2}_L= \left(
            \begin{array}{ccc|ccc}
              {\bf 0}_{2k} & {\bf 0} & {\bf 0} & {\bf 0} & I_k & (U^t)^{-1}+I_k \\
              {\bf 0}& \left[\left(U^t\right)^{-1}\right]^{t}+I_k & I_k & {\bf 0} & {\bf 0} & {\bf 0} \\
              {\bf 0}& {\bf 0} &{\bf 0}& I_{2k}& {\bf 0}& {\bf 0}
            \end{array}
          \right),
\eeq
and
\beq
\Omega^{C_3}_L=\left(
\begin{array}{ccc|ccc}
I_k & {\bf 0} & {\bf 0 } & {\bf 0}_k& {\bf 0} & {\bf 0} \\[2pt]
{\bf 0} & {\bf 0}_k & {\bf 0} & {\bf 0}  & I_k  & {\bf 0} \\[2pt]
{\bf 0}& {\bf 0} & I_{2k} & {\bf 0} & {\bf 0 }& {\bf 0}_{2k}\\
\end{array}
\right).
\eeq
$\left|\Omega_L^{C_1}\right\>$ is actually a $k$-fold tensor product of Bell states. $\left|\Omega_L^{C_2}\right\>$ is the key resource state in our procedure to reduce the depth of -C- stage computation. All the ancillas here are CSS states.
The net effect is the desired \text{-C-} stage acting on $\text{Q}_{1}, \dots, \text{Q}_k$, and the auxiliary qubits $\text{A}_{1}, \dots, \text{A}_k$ are reset to $|+_L\>^{\otimes k}$ (up to logical $Z$ corrections). One can transform $\text{A}_{1}, \dots, \text{A}_k$ back into $|0_L\>^{\otimes k}$ or just keep them and start with $|+_L\>^{\otimes k}$ for the next stage. The procedure with $\text{A}_{1}, \dots, \text{A}_k$ initially in the $|+_L\>^{\otimes k}$ state for -C- stage is similar. As a conclusion, one has the following theorem:
\begin{mytheorem}
For an $\llb n, 2k, d \rrb$ CSS code, any logical Clifford circuit on $k$ logical qubits can be realized fault-tolerantly by 22 rounds of single-shot Steane syndrome measurement.
\end{mytheorem}

%
%

\subsection{FT preparation of qualified ancilla states }
The ancilla states listed in the last subsection can be prepared fault-tolerantly by using Shor syndrome measurement to measure all stabilizer generators and logical Pauli operators, which will be discussed in Sec.~\ref{sec:related_FTQC}.
In this subsection, we generalize the FT state preparation protocol in \cite{Ancilla_distillation_1, zheng2017efficient} to show that all the logical stabilizer ancilla states required in the previous subsection can be prepared fault-tolerantly by distillation with almost constant overhead in terms of the number of qubits.

Since all logical ancilla states we considered are stabilizer states, once the eigenvalues of all their stabilizers are known, one can remove the errors completely. The basic idea of distillation is shown in Fig.~\ref{fig:ft_distillation}---many copies of the logical ancilla states are prepared in the physical level via the same $U_{\text{prep}}$ and then encoded by $U_{\text{enc}}$ to the same large block code. Both $U_{\text{prep}}$ and $U_{\text{enc}}$ are noisy in practice. They are sent into a distillation circuit (which is also noisy) and certain blocks are measured bitwise. The eigenvalues of all the stabilizers in group $\mathcal{S}$ and the logical Pauli operators of the output blocks can be estimated if the distillation circuit is constructed based on the parity-check matrix of some classical error-correcting code---the flipped eigenvalues caused by errors during state preparation and distillation can be treated as classical noise and thus be decoded. Since correlated errors remaining on a block can cause its estimated eigenvalues not compatible with each other, each output block needs further check for compatibility. Postselection (rejecting the blocks whose estimated eigenvalues are incompatible) is then done to remove the blocks likely containing correlated errors. Error correction is then applied based on the estimated eigenvalues of stabilizers to the remaining blocks.
\begin{figure}[!htp]
\centering\includegraphics[width=90mm]{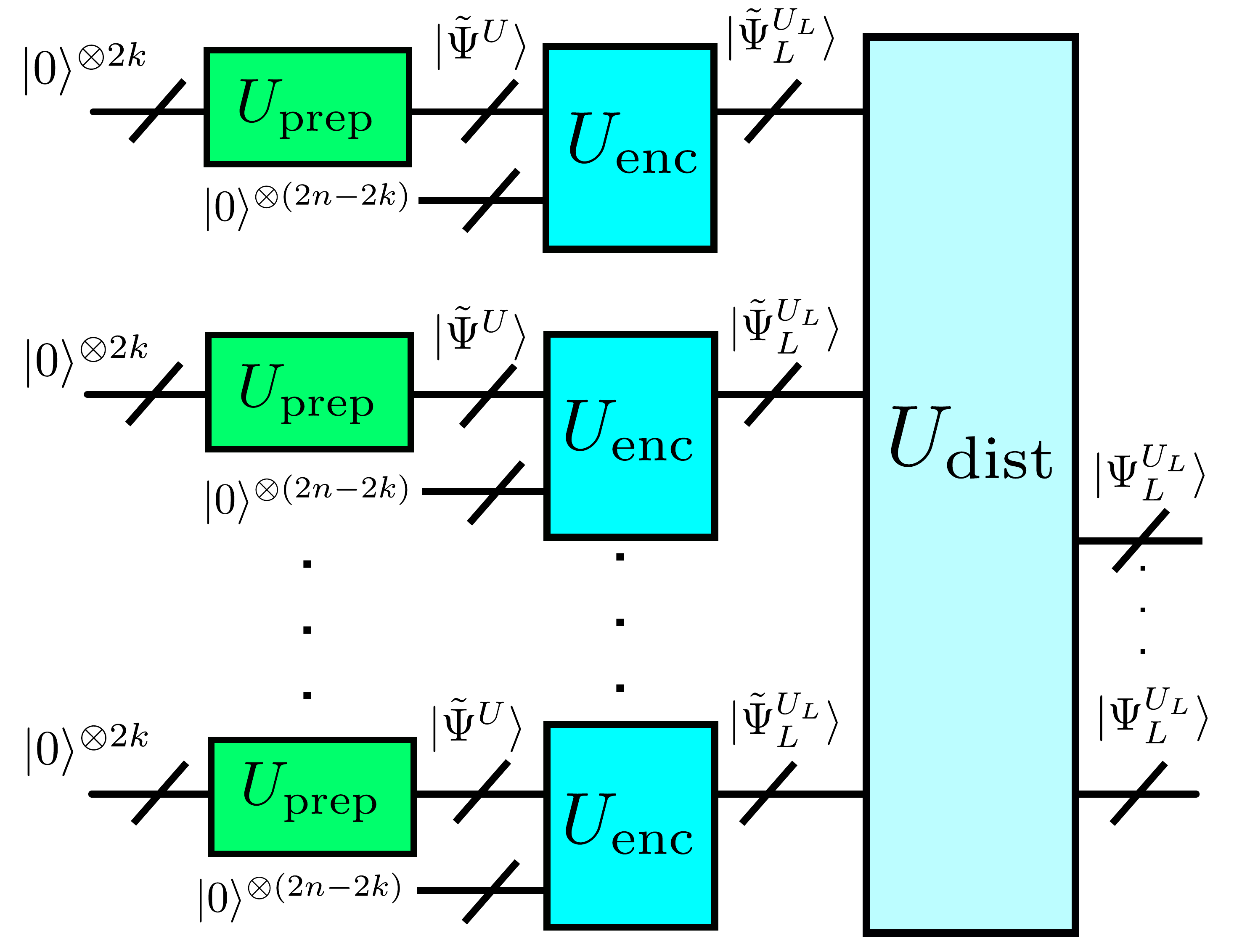}
\caption{\label{fig:ft_distillation}FT ancilla state distillation circuit for $\left|\Psi_L^{U_L}\right\>$.  Ancilla states are prepared and encoded via noisy quantum circuits $U_{\text{prep}}$ and $U_{\text{enc}}$. Then they are fed to the distillation circuit and output the qualified ancilla state. The preparation of $|\Omega^{\bf EF}_{L}\>$ is similar.
}
\end{figure}

The distillation circuit can be synthesized according the parity-check matrix of an $[n_c,k_c,d_c]$ classical code, which has the form $\textsf{H}=(I_{n_c-k_c}| \ \textsf{A}_c)$. Consider a group of $n_c$ ancilla blocks. Choose the first $r_c=n_c-k_c$ ancillas blocks to hold the classical parity checks, and do transversal CNOTs from the remaining $k_c$ ancillas onto each of the parity-check ancillas according to the pattern of 1s in the rows of $\textsf{A}_c$: if $[\textsf{A}_c]_{i,j}=1$, we apply a transversal CNOT from the $(r_c+j)$th ancilla to the $i$th ancilla block. Then measure all the qubits on each of the first $r_c$ ancilla blocks in the $Z$ basis, which destroys the states of those blocks and extracts information to estimate the eigenvalues of all $Z$ types stabilizers and logical operators of the remaining $k_c$ blocks. In the low error regime, after filtering out the blocks with incompatible estimated eigenvalues of stabilizers, the output blocks will contain no correlated $X$ errors after subsequent error correction if $d_c$ is larger than the distance of the underlying CSS code~\cite{zheng2017efficient}. Correlated $Z$ errors can be removed in a similar manner.

\begin{figure}[!htp]
\centering\includegraphics[width=65mm]{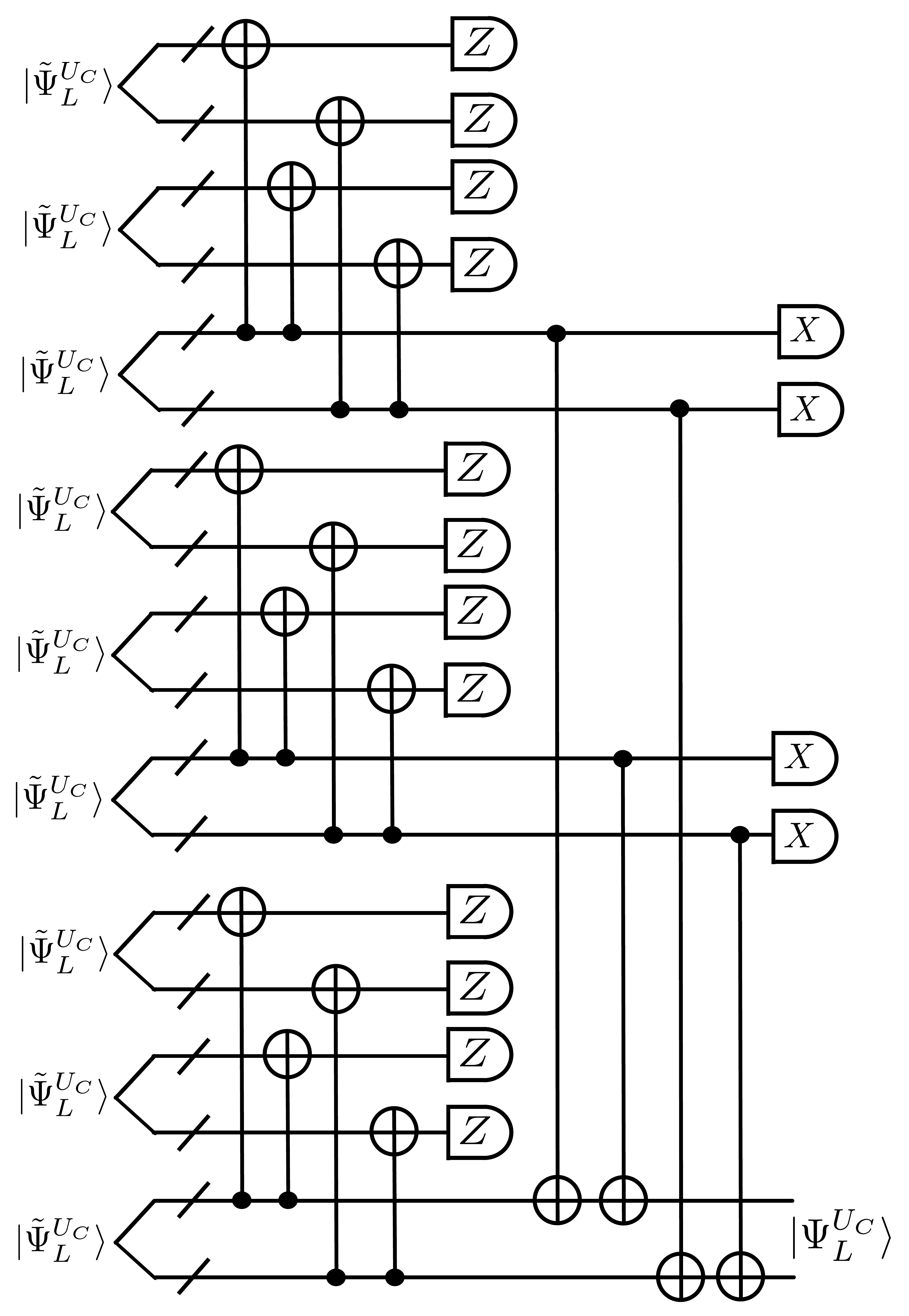}
\caption{\label{fig:ft_distillation_CSS} Two-stage FT ancilla state distillation circuit for $|\Psi_L^{U_C}\>$. The distillation circuit is based on the parity-check matrix of the $[3,1,3]$ code for both stages. Actually this circuit can distill any logical CSS states fault-tolerantly.
}
\end{figure}

One can concatenate two stages of distillation (with the output blocks of the first stage randomly shuffled) to remove both correlated $X$ and $Z$ errors. Fig.~\ref{fig:ft_distillation_CSS} shows the overall circuit to distill $\left|\Psi_L^{U_C}\right\>$ using the $[3,1,3]$ code for both stages. For a two-stage distillation protocol based on two classical $[n_{c_1},k_{c_1},d_{c_1}]$ and $[n_{c_2},k_{c_2},d_{c_2}]$ codes, the number of input and output blocks are $n_{c_1}n_{c_2}$ and $Y(p)\cdot n_{c_1}n_{c_2}$ respectively. Here, $Y(p)$ is the yield rate defined as
\beq
Y(p)=
\frac{k_{c_1}k_{c_2}(1-R_1(p))(1-R_2(p))}{n_{c_1}n_{c_2}},
\eeq
where $R_i(p)$ is the block rejection rate for postselection in the $i$th stage of distillation for a gate/measurement error rate $p$. Asymptotically, the rejection rate for each round of distillation is $O(p^2)$, because at least two failures are needed to cause a rejection of the output blocks. Thus, it is likely that $R_1(p)$ and $R_2(p)$ negligible in the small $p$ regime. On the other hand, good capacity-achieving classical codes exist such that $\frac{k_{c_1}k_{c_2}}{n_{c_1}n_{c_2}}$ can be independent of the code distance of the underlying CSS code to ensure $d_c>d$, so that the distillation circuits are still able to output qualified ancilla states. Hence, $Y(p)$ can achieve almost $\Theta(1)$ for sufficiently low $p$.

As shown in the last subsection, one needs several ancilla states, which can be grouped into three types:

---\emph{Type I}--- These are CSS states including $\left|\Psi_L^{U_C}\right\>$, $\left|\Omega_L^{P_2}\right\>$ $\left|\Omega_L^{H_2}\right\>$, $\left|\Omega_L^{C_1}\right\>$, $\left|\Omega_L^{C_2}\right\>$, and $\left|\Omega_L^{C_3}\right\>$. They are stabilized by logical operators which are tensor products of either $X$ or $Z$.  They can be prepared and distilled by the circuit in Fig.~\ref{fig:ft_distillation_CSS}: eigenvalues of the $Z$ ($X$) stabilizer generators and $Z$ ($X$) logical operators are checked at the first (second) round to remove correlated $X$ ($Z$) errors.

---\emph{Type II}--- These are CSS states up to logical Hadamard gates applied to some logical qubits, including $\left|\Psi_L^{U_H}\right\>$ and $\left|\Omega_L^{H_1}\right\rangle$. In this case, we distill the upper block to remove correlated $X$ errors and the lower block to remove correlated $Z$ errors in the first round and reverse the order in the second round, as shown in Fig.~\ref{fig:ft_distillation_H}.

\begin{figure}[!htp]
\centering\includegraphics[width=65mm]{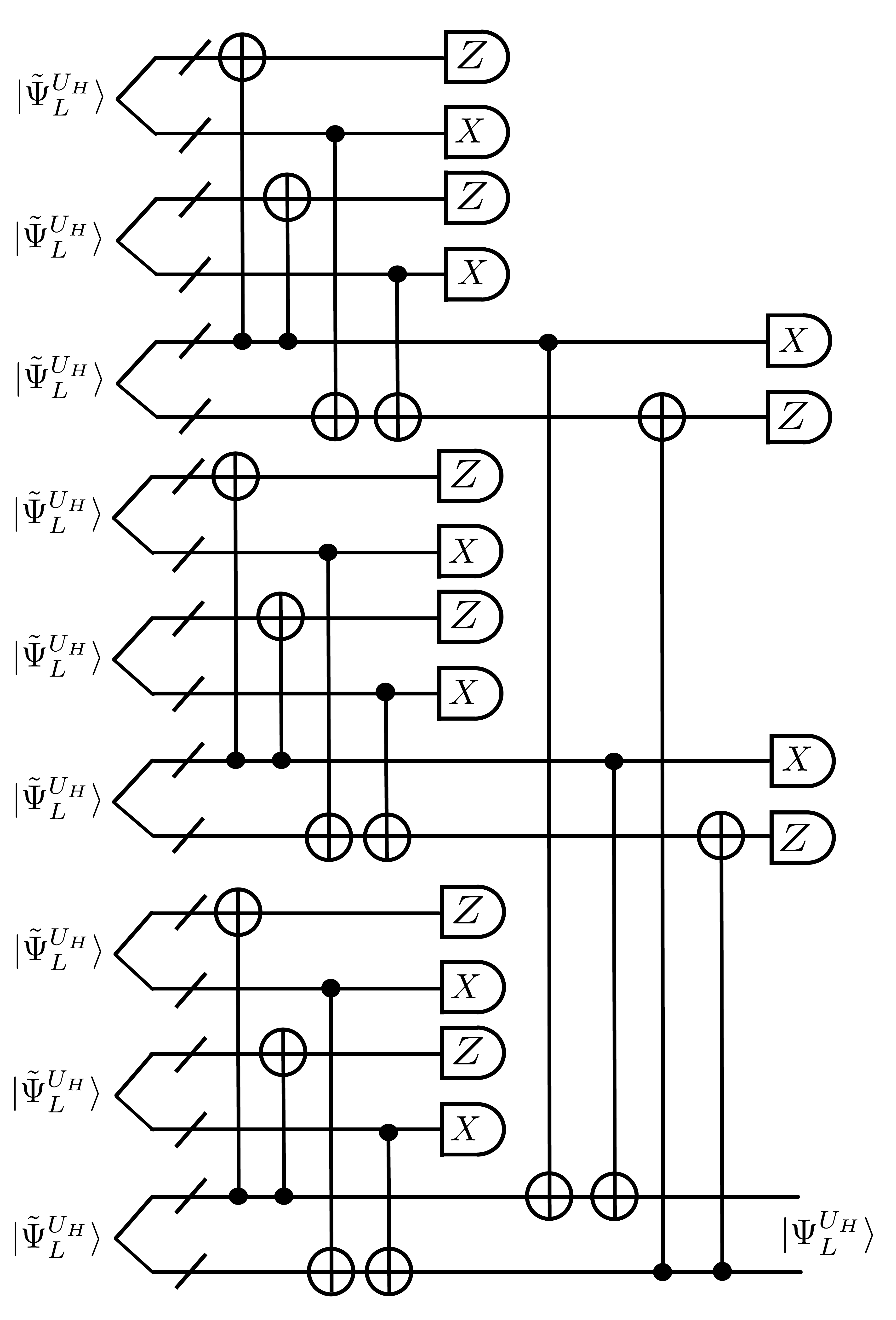}
\caption{\label{fig:ft_distillation_H} Two-stage FT ancilla state distillation circuit for $\left|\Psi_L^{U_H}\right\>$ based on the parity-check matrix of $[3,1,3]$ code. This circuit can be used to prepare any logical CSS states up to logical Hadamard gates fault-tolerantly.
}
\end{figure}

---\emph{Type III}--- This set of states are Type I or II states up to logical Phase gates applied to some logical qubits, including $\left|\Psi_L^{U_P}\right\rangle$ and $\left|\Omega_L^{P_2}\right\rangle$. We will confine our attention to doubly even and self-dual CSS codes and set the weight of logical $X$ operators $X_{j,L}$ to odd numbers for all $j$.

For $\left|\Psi_L^{U_P}\right\rangle$, one could first prepare a qualified CSS state of the form
\beq
\Psi_L^{U_P'}=\left(
  \begin{array}{cc|cc}
    I_k & \Lambda_k^m & {\bf 0} & {\bf 0} \\[2pt]
    {\bf 0} & {\bf 0} &  \Lambda_k^m & I_k\\
  \end{array}
\right)
\eeq
via the distillation circuit in Fig.~\ref{fig:ft_distillation_CSS}.
Then, Phase gates are applied bitwise to the lower ancilla block, which will transform the state to
\beq
\Psi_L^{U_P''}=\left(
  \begin{array}{cc|cc}
    I_k & \Lambda_k^m & {\bf 0} & \Lambda_k^m \\[2pt]
    {\bf 0} & {\bf 0} &  \Lambda_k^m & I_k\\
  \end{array}
\right).
\eeq
This is because the bitwise Phase gates will preserve the stabilizer group while implementing logical Phase gates on all the logical qubits. Then one can apply logical CNOTs (assisted by some CSS states) from the upper block to the lower block on the remaining $k-m$ logical qubits to obtain a qualified $\left|\Psi_L^{U_P}\right\>$.

Similarly, for $\left|\Omega_L^{P_1}\right\>$, one can  prepare a qualified Type II state of the form
\beq
\Omega^{P'_1}_L=\left(
  \begin{array}{cccc|cccc}
     I_k & {\bf 0} & {\bf 0} & {\bf 0} & {\bf 0} & {\bf 0} & {\bf 0} & {\bf 0} \\
     {\bf 0} &\Lambda_k^m &{\bf 0} &{\bf 0} &{\bf 0} & I_k+\Lambda_k^m &{\bf 0} & {\Lambda}_k^m \\
    {\bf 0} & {\bf 0} & I_k & {\bf 0} & {\bf 0}& {\bf 0} & {\bf 0}_k & {\bf 0}\\
    {\bf 0}& {\bf 0} & {\bf 0} & I_k & {\bf 0} & \Lambda_k^m & {\bf 0}&  {\bf 0}_k\\
  \end{array}
\right)
\eeq
by the circuit in Fig.~\ref{fig:ft_distillation_H}. After that, bitwise Phase gates are applied to the upper block to transform $\left|\Omega^{P'_1}_L\right\>$ to
\beq
\Omega^{P''_1}_L=\left(
  \begin{array}{cccc|cccc}
     I_k& {\bf 0} & {\bf 0} & {\bf 0} & {\bf 0} & {\bf 0} & {\bf 0} & {\bf 0} \\[2pt]
      {\bf 0}&\Lambda_k^m &{\bf 0} &{\bf 0} & {\bf 0} & I_k &  {\bf 0}& {\Lambda}_k^m \\[2pt]
    {\bf 0} & {\bf 0} & I_k & {\bf 0} & {\bf 0}& {\bf 0} & {\bf 0}_k & {\bf 0}\\[2pt]
    {\bf 0}& {\bf 0} & {\bf 0} & I_k & {\bf 0} & \Lambda_k^m & {\bf 0}&  {\bf 0}_k\\
  \end{array}
\right).
\eeq
We can then measure the operators $({\bf 0} \ \ {\bf 0} \ | \ I_k \ \Lambda_k^m )$ (assisted by some CSS states) on the upper block and obtain
\beq
\Omega^{P_1}_L=\left(
  \begin{array}{cccc|cccc}
    {\bf 0} & {\bf 0} & {\bf 0} & {\bf 0} & I_k &  \Lambda_k^m & {\bf 0} & {\bf 0} \\[2pt]
     \Lambda_k^m &\Lambda_k^m &{\bf 0} &{\bf 0} & {\bf 0}  & I_k & {\bf 0} & {\Lambda}_k^m \\[2pt]
    {\bf 0} & {\bf 0} & I_k & {\bf 0} & {\bf 0}& {\bf 0} & {\bf 0}_k & {\bf 0}\\[2pt]
    {\bf 0}& {\bf 0} & {\bf 0} & I_k & {\bf 0} & \Lambda_k^m & {\bf 0}&  {\bf 0}_k\\
  \end{array}
\right)
\eeq
up to logical Pauli corrections.

\subsection{Resource overhead}
We estimate the average number of qubits and physical gates to implement a logical Clifford circuit in this subsection.

For both Knill and Steane syndrome measurements, one only needs a constant  rounds of circuit teleportations or Pauli measurements. Asymptotically, the fault-tolerant ancilla state distillation protocol dominates the resource cost. We can recycle the ancilla qubits after they are used to further reduce the redundancy. However, it will not change the asymptotic scaling of resource cost.

As discussed in the previous section, the overall number of qubits required for a two-stage state distillation protocol based on the parity-check matrices of two classical codes is
$$N_q = c n  n_{c_1} n_{c_2},$$
where $c$ is some constant. The numbers of gates for each $U_{\text{prep}}$ and $U_{\text{enc}}$ in Fig.~\ref{fig:ft_distillation} are $O(k^2/\log k)$ and $O(n^2/\log n)$, respectively. Therefore, the total number of gates for noisy logical stabilizer state preparation at the physical level is
$$
N_{\text{enc}}= \frac{c n^2 n_{c_1}n_{c_2} }{\log n}
$$
with depth $O(n)$.

The number of gates for the two-round distillation circuit depends on $\textsf{A}_{c}$, which has $O(k_c^2)$ 1s. Hence, the number of gates required for $U_{\text{dist}}$ for two stages is
$$
N_{\text{dist}} = c_1 k_{c_1}^2n n_{c_2} + c_2 k_{c_2}^2 n,
$$
with depth $O\left(\max\{k_{c_1},k_{c_2}\}\right)$, where $c_1,c_2$ are  constants. Therefore, the overall  depth of ancilla state preparation is
$$
O(\max\{n,k_{c_1}, k_{c_2}\}).
$$

Note that there are $Y(p) n_{c_1}n_{c_2}$ output ancilla blocks per round of distillation and hence the same number of identical circuits can be implemented. On average, one needs
$$
\bar{N}_q=\frac{cn_{c_1}n_{c_2}n}{Y(p)n_{c_1}n_{c_2}}\sim O(n)
$$
qubits for a logical Clifford circuit. On the other hand, the physical gates required for raw state preparation is
$$
\bar{N}_{\text{enc}}=\frac{c n^2 n_{c_1}n_{c_2} }{Y(p)n_{c_1}n_{c_2}\log n}\sim O(n^2/\log n).
$$
If two good classical capacity-achieving codes are used(e.g., low-density parity-check (LDPC) codes~\cite{MacKay:2003:CambridgeUniversityPress}) with $k_{c_i}/n_{c_i}=\Theta(1)$, $i=1,2$, and if we restrict ourselves to $k_{c_1}\lesssim n/\log n$ and $k_{c_2}/n_{c_1}=\Theta(1)$, the average number of physical gates required for distillation will be
$$
\bar{N}_{\text{dist}}=\frac{c_1 k_{c_1}^2 n n_{c_2} + c_2 k_{c_2}^2 n}{Y(p)n_{c_1}n_{c_2}}\sim O(n^2/\log n).
$$
To sum up, an arbitrary logical Clifford circuit requires
$\bar{N}_{\text{gate}} = O(n^2/\log n)$
physical gates on average. If one considers the family of large CSS block codes (e.g. quantum LDPC codes) with $k/n\sim \Theta(1)$, only $O(k)$ physical qubits and $O(k^2/\log k)$ physical gates are needed on average to implement any logical Clifford circuit, with off-line circuit depth $O\left(\max\{k,k_{c_1}, k_{c_2}\}\right)$ for ancilla preparation. These results suggest that the numbers of qubits and gates required for logical Clifford circuits have the same scaling as the physical level, when the physical error rate is sufficiently low and good classical LDPC codes are used.

\begin{remark}
One can achieve very high efficiency of resource utilization with our scheme in the following scenario: a small batch of finite Clifford circuits are repeatedly applied in an algorithm for a certain period of time. This is because our distillation process has high throughput --- it can produce a large number of \emph{identical} ancilla states (and thus generate the same number of identical Clifford circuits), and each state preparation is efficient.
Several questions are raised here naturally: is there any useful quantum algorithm whose quantum circuits have such structure? In other words, is there an ansatz to adapt our FTQC architecture to design a circuit for some particular algorithm? Is there a good computation architecture to efficiently generate the large amount of identical ancilla states? These questions are all open and need further investigation.
One promising candidate here is the Hamiltonian simulation algorithm for quantum simulations~\cite{Lloyd:1996:1073,aspuru2005simulated, wecker2014gate,Hastings:2015_simulation,Poulin:2015qic_simulation,garnet2020quantum}. In that case, the target Hamiltonian changes slowly during the computation. In a certain period of time, the Trotter decomposition can be regarded as identical. Another candidate is the optimization type algorithms like
Quantum Approximate Optimization Algorithm (QAOA)~\cite{farhi2014quantumqaoa}, which needs to rapidly apply Hadamards.
On the other hand, to generate large mount ancilla states, Single Instruction Multiple Data (SIMD) style architecture~\cite{heckey2015compiler, risque2016characterization} that apply the same quantum gates on multiple qubits in the same region simultaneously, maybe particularly useful.
\end{remark}

\subsection{Summary}
In conclusion, we provided two methods implementing logical Clifford circuits fault-tolerantly via constant number of steps of Knill or Steane syndrome measurement circuits \emph{in-situ}. Each method requires certain types of logical stabilizer states as ancillas. We showed that all ancilla states listed can be prepared fault-tolerantly through two-stage distillation.
Our method transfers the complexity of Clifford circuits on logical level to the complexity of state preparation $U_{\text{prep}}$ on physical level completely, which can be done offline. Surprisingly, if one chooses large block codes with encoding rate $k/n\sim \Theta(1)$, the overall numbers of qubits and physical gates required for a Clifford circuit on $\llb n,k,d\rrb$ CSS circuit are around $O(k)$ and $O(k^2/\log k)$, respectively, which are independent of the distance of the underlying CSS code. This is the same scaling as a perfect Clifford circuit acting on $k$ physical qubits.

\section{Discussion}\label{sec:discussion}
In this section, we compare the method proposed in this paper with some other related fault-tolerant protocols in the literature including one-way quantum computation. Then we estimate the numbers of physical qubits and gates required for each scheme. The results are summarized in Table~\ref{table_resoure}. Since different FTQC schemes have different working regions and performance, these results only provide a rough insight of resource scaling. We also discuss the potential improvements on ancilla state preparation for further reduce the overhead.

\begin{table*}[tp!]
{\footnotesize
\begin{tabular}{c|c|c|c|c|c}
  \hline
  \hline
  Method/Average resource cost &\, \#. physical qubits \, &\, \#. physical operatation \,& \, \#. ancilla states \, & \,  \emph{in-situ} depth & off-line depth \rule{0pt}{2.6ex}\,\\[2pt]
   \hline
  \rule{0pt}{2.6ex}
     Standard circuit model & $O(k)$ & $O(k^2/\log k)$ & N/A & $O(k)$ & N/A\\[2pt]
   \emph{Circuit model FTQC (this paper)}\, & $O(k)$  & $O(k^2/\log k)$  & $O(1)$   & $O(1)$   & $O(\max\{k,k_{c_1},k_{c_2}\})$ \\[2pt]
   Circuit model FTQC (as in Ref.~\cite{Gottesman:1999:390, Zhou:2000:052316}) &  $O(\max \{k, wd \})$ & $O(\max\{k w d, k^2/\log k\})$ &  $O(kd)$ & $O(kd)$  & $O(kd)$ \\[2pt]
   Circuit model FTQC  (as in Ref.~\cite{steane1999efficient_Nature,steane2005fault,brun2015teleportation}) &  $O(k)$ & $O\left(k^4/(\log k)^2\right)$ &  $O(k^2/\log k)$   & $O(k^2/\log k)$ &  $O(\max\{k,k_{c_1},k_{c_2}\})$ \\[2pt]
   one-way QC (as in Ref.~\cite{Raussendorf:2003:022312})  & $O(k^3/\log k)$  & $O(k^3/\log k)$  & N/A  &   $O(1)$  & N/A \\[2pt]
   FT one-way QC (as in Ref.~\cite{Raussendorf:2006:2242,Raussendorf:2007:199}) & $O\left(k^3 d^3 / \log k \right)$ & $O\left(k^3 d^3 / \log k \right)$  & N/A  & $O(1)$&  N/A\\[2pt]
   Surface code (as in Ref~\cite{Folwer2012PhysRevA.86.032324}) & $O(kd^2)$ & $O(k^3d^3/\log k)$  & N/A  & $O(k^2d/\log k)$  & N/A\\[3pt]
  \hline
  \hline
\end{tabular}
}
\caption{\label{table_resoure} Resources required for a Clifford circuit on $k$ logical qubits at the physical and logical levels for an $\llb n,k, d \rrb$ CSS code. The physical operations counts all state preparation, gates, measurements including ancilla preparation and error correction. We assume $k/n\sim \Theta(1)$ for the large block codes and $w$ is the maximum weight of the stabilizers. Sufficient parallelization are also considered to minimize the depth.}
\end{table*}

\subsection{Related FTQC protocols}\label{sec:related_FTQC}

In this paper, we implement FT logical circuit teleportation via the
single-shot Knill syndrome measurement protocol and FT ancilla distillation. It is worthwhile to compare this with the original teleportation-based FTQC in Ref.~\cite{Gottesman:1999:390,Zhou:2000:052316}. Rather than Knill syndrome measurement, Shor syndrome measurement~\cite{Shor:1996:56} is used for error correction and ancilla state preparation. Our construction of logical Clifford circuits through a constant number of steps of teleportation is also possible in that scenario, where the ancilla state preparation again dominates the resource cost.

In Ref.~\cite{Gottesman:1999:390,Zhou:2000:052316}, $O(1)$ logical ancilla states of size $O(n)$ is required. To prepare qualified logical ancilla states, one applies Shor syndrome measurement to measure $n$ stabilizers, including stabilizer generators and logical operators. Each measurements needs one cat state. Each cat state consists $O(w)$ qubits, which takes $O(w)$ CNOTs to prepare, where $w$ is the maximum weight of the stabilizers. A verification is also required after the raw preparation of each cat state, which also takes $O(w)$ CNOTs and rejects the states with probability around $O(p)$. Transversal CNOTs from the verified cat state to the code block are then applied to extract the eigenvalues of the stabilizers, which takes $O(w)$ CNOTs. To establish reliable eigenvalues for the stabilizers of an $\llb n,k,d\rrb$ code, $O(d)$ rounds of Shor syndrome measurements and a majority vote are required for each stabilizer.  Thus the overall number of physical gates required for state preparation is $O(\max\{nwd, n^2/\log n\})$ with depth $O(nd)$.

For large block codes with $k/n\sim \Theta(1)$, the number of physical gates required for ancilla state preparation is $O(\max\{kwd, k^2/\log k\})$ with depth $O(kd)$. Meanwhile, $O(1)$ logical ancilla states are required, which needs $O(k)$ physical ancilla qubits altogether.  It also takes $O(kd)$ rounds of serial Shor syndrome measurements (since the stabilizers are in general highly overlapped) to do error correction on the data block, each round consumes a verified cat states. Hence, the depth for a logical Clifford circuit is $O(kd)$ and the same number of verified cat states are needed. Assuming the qubits supporting cat states are recycled after they are measured, one needs $O(wd)$ ancilla qubits for cat states throughout the process after parallelization. The number of all ancilla qubits is thus $O(\max\{k, wd\})$. This way of implementing logical Clifford circuits needs more physical gates when $w$ is large and takes a much longer computation time for large $k$ and $d$.

Our scheme also greatly simplifies the block-code based FTQC using Steane syndrome measurement in Ref.~\cite{steane1999efficient_Nature,steane2005fault,brun2015teleportation}. There, logical Clifford gates are implemented one by one, and thus $O(k^2/\log k)$ different ancilla states of size $2n$ qubits are required and $O(k^2/\log k)$ rounds of Steane syndrome measurements are needed. With the same ancilla distillation protocol and ancilla recycling, one needs $k^4/(\log k)^2$ gates and $O(k)$ qubits for every single circuit on average for finite rate codes with $k/n\sim \Theta(1)$.


%
%
%
%
%
%

\subsection{One-way quantum computing}
For one-way QC, one initially prepares a cluster state consisting of a large number of qubits. Quantum information is then loaded onto the cluster and processed through single-qubit measurements on the cluster state substrate. It can be shown that all quantum circuits can be mapped to the form of one-way QC. In general, one-qubit measurements are performed in a certain
temporal order and in a spatial pattern of adaptive measurement bases based on previous measurement outcomes.
Interestingly, for those qubits supporting Clifford circuits,
no measurement bases have to be adjusted (i.e., those of
which the operator $X$, $Y$ or $Z$ is measured). Thus, for any given quantum circuit, all of its Clifford gates can be realized
\emph{simultaneously} in the first round of single-qubit measurements, regardless of their space-time locations in the circuit~\cite{Raussendorf:2003:022312}, if a sufficiently large cluster state is provided to support the whole computation circuit. Specifically, for a cluster state in 2D, it requires $O(k^3/\log k)$ supporting cluster qubits and single-qubit measurements for an instantaneous Clifford circuit without error correction.

However, it is difficult to control errors if a cluster state large enough to support the entire quantum computation is used, since the computation might reach certain qubits only after a long time, so that these qubits would already suffer significant errors. This scheme is not fault-tolerant. By contrast, if the computation is split, then the size of sub-circuits may be adjusted so that each of them can be performed within some constant time. The measured qubits are then recycled to entangle with the unmeasured qubits to form a new cluster for the next computation step. In this way, each cluster qubit is exposed to constant decoherence time before being measured and the error rate is bounded. In this case, FT one-way QC is possible~\cite{Raussendorf:2003:022312}. Note that it is still possible to perform a Clifford circuit during the computation in one time step, if  $O(k^3/\log k)$ qubits are provided at the same time, but it is no longer possible to finish all the Clifford gates in the computation in a single time step.

To complete the discussion, here we consider FT one-way QC in 3D lattice as in Ref.~\cite{Raussendorf:2006:2242, Raussendorf:2007:199}. The Clifford circuits are performed through single-qubit measurements in the $Z$ basis to create topologically entangled defects in the 3D lattice. The remaining qubits are measured in the $X$ basis to provide syndrome information for 3D topological error correction~\cite{raussendorf2005long}. For $k$ encoded qubits in \cite{Raussendorf:2007:199} with distance $d$ boundary surface codes, the number of cluster qubits needed in a single 2D slice is $O(kd^2)$ and it requires $O(k^2/\log k)$ slices for an arbitrary Clifford circuit in the worst case. Thus, it takes the volume of a cluster state comprising $O(k^3d^3/\log k)$ qubits and the same number of single-qubit measurements. As a comparison, the variants of FT one-way QC in 2D based on the surface code~\cite{Folwer2012PhysRevA.86.032324} need $kd^2$ physical qubits for encoding and each logical CNOT gate takes $O(d)$ time steps.

In conclusion, even though one-way QC can in principle implement the Clifford gates of a circuit in a single time step, it requires many more physical qubits, whether it is implemented in a fault-tolerant manner or not. It is worth noting that the FT one-way QC and its 2D variants require only local operation, which is a great practical advantage, since the codes considered in our scheme are highly non-local in general. However, our results suggest the potential for huge resource reduction for FTQC if non-local operations are allowed.

\subsection{More efficient ancilla state preparation}
The distillation protocol mentioned in this paper is basically the same as the one in Ref.~\cite{zheng2017efficient}. The main difference is that the ancilla states distilled here can generate a whole circuit rather than a single gate on the data code block. These will cause an extra complexity on $U_{\text{prep}}$ stage in Fig.~\ref{fig:ft_distillation} up to $O(k^2 /\log k)$ gates with an extra depth $O(k)$. Meanwhile, the overall number of gates and depth for the whole distillation protocol ($U_{\text{prep}}$, $U_{\text{enc}}$ and $U_{\text{dist}}$ combined) also scale as $O(k^2/\log k)$ and $O(k)$. Note that the extra depth of ancilla preparation is negligible if they are produced in a pipeline manner. Thus, in the worst case, it will cause a constant decrease of distillation quality. But in practice, $U_{\text{prep}}$ may only take a small portion of the whole protocol. On the other hand, since we generate a circuit rather than a single gate at one time, it will reduce the quality requirement of the output ancilla states to support FTQC, and hence the error rate requirement for each operation as well. The overall effect of extra preparation complexity on distillation needs further exploration.

The two-stage distillation protocol gives $O(1)$ yield rate on average. However, the number of input ancilla blocks required simultaneously is $O(n_{c_1}n_{c_2})$, which can be huge in practice. Consequently, the distillation circuit is still relatively complicated and error can occur in many positions. As a result, the typical acceptable error rate (or threshold for distillation) is less than $10^{-4}$~\cite{zheng2017efficient}, which is challenging with current technologies like superconducting qubits~\cite{supremacy}. Meanwhile, in many cases, one doesn't need as many as $O(k_{c_1}k_{c_2})$ identical ancilla states to generate that large number of the same Clifford circuits.

There are two ways to further simplify the distillation process and reduce overall number of qubits: in stead of two-stage distillation, one can filter out one type of correlated errors at the beginning by single-round stabilizer measurements with Steane Latin rectangle method~\cite{steane2002fast}, which takes advantage of the fact that only a small set of  correlated errors needs to be removed according to the degeneracy of quantum code. Then we remove the other type of correlated errors through distillation. The depth of such preparation is also $O(\max\{k,k_c\})$
Here, only $O(n_{c})$ input code blocks are required simultaneously and it generates $O(k_{c})$ identical output states. It not only reduces the complexity of preparation circuit but also gives more flexibility. The second method is to take advantage of the symmetry of the underlying CSS codes: the qubits of different code blocks are permuted in different ways after raw preparation (though finding such permutation may be challenging), so that the correlation of errors between code blocks can be suppressed~\cite{paetznick_ben2011fault}. Consequently, it may require less input code blocks for distillation. In principle, these two methods can also be combined together and their effect needs further investigation.

\acknowledgments
The funding support from the National Research Foundation \& Ministry of Education, Singapore, is acknowledged. This work is also supported by the National Research Foundation of Singapore and Yale-NUS College (through grant number IG14-LR001 and a startup grant).
CYL was supported by the Ministry of Science and Technology, Taiwan under Grant MOST108-2636-E-009-004.
TAB was supported by NSF Grants No. CCF-1421078 and No. MPS-1719778, and by an IBM Einstein Fellowship at the Institute for Advanced Study.
\appendix

\vspace{2mm}

\section{Proof of Lemma~\ref{lemma:stabilizer_transform}}\label{sec:proof_lemma}
\begin{proof}
Let $l'_j$ denote the $j$th row vector of $L'$ and $c_p$ be the $p$th column vector of $(I_n \ I_n)^t$.
Equation~(\ref{eq:orthogonal}) is equivalent to
\beq
l'_j c_p = \delta_{jp}, \quad \quad  1\leq j, p\leq n, \label{eq:orthogonal2}
\eeq
where  $\delta$ is the Kroneker delta function.

Let $l_j$ denote the $j$th row vector of $L$. Obviously, $l_1=(1,0,\dots, 0)$, satisfying
$l_1 c_p = \delta_{1p}$. Let $l'_1=l_1$.

It is easy to see that $l_j c_p = 0$ for $p>j$, since $L_1$ is a lower triangular matrix. With all the diagonal elements of $L_1$ being~0, one has
\beq
l_jc_j=1.\label{eq:lc}
\eeq
Define the set $\mathscr{I}_j=\{p\ |\ l_{j}c_p =1, p < j \}$.
For $j=2,\dots,n$, let
\beq
l_{j}'=l_{j} + \sum_{p\in \mathscr{I}_j} l_p'. \label{eq:lp}
\eeq
We also define a matrix $L'^{(j)}$ that contains the rows $l_1',\dots,l_j'$: 
\beqs
L'^{(j)}=\left(
           \begin{array}{c}
             l_1' \\[2pt]
             \vdots \\[2pt]
             l_j' \\
           \end{array}
         \right).
\eeqs
Since $L_1$ is lower triangular, and the summation of  $l_p$  in Eq.~(\ref{eq:lp}) only counts the terms with  $p<j$, $L'^{(j)}$ can be written as
\beqs
L'^{(j)}=\left(L_2^{(j)}\ L_3^{(j)}\right),
\eeqs
where $L_2^{(j)}$ and $L_3^{(j)}$ are also lower triangular matrices. Eventually, we have $L_2=L_2^{(n)}$ and $L_3=L_3^{(n)}$.

It remains to prove Eq.~(\ref{eq:orthogonal2}).
We prove this by induction.
For $j = 2$, if $l_2 c_1 = 1$, one has $l_2'=l_2 + l_1$. Thus $l_2' c_1 = 0$ and $l_2'c_2=1$, since $l_1c_1=1$ and $l_1c_2=0$. Also, $l_2' c_p=0$ for $p>2$ since $L_2'^{(2)}$ and $L_3'^{(2)}$ are lower triangular matrices. So $l'_2 c_p = \delta_{2p}$ holds for $1\leq p \leq n$.

Now assume $l'_{1}c_p=\delta_{1p}$, $\dots,$ $l'_{j}c_p=\delta_{jp}$ holds. 
Then
\beqs
l_{j+1}'c_q=l_{j+1}c_q + \sum_{p\in \mathscr{I}_{j+1} } l_p'c_q.
\eeqs
Consider $q < j+1$ first. If $l_{j+1}c_q=1$, then $q\in \mathscr{I}_{j+1}$ and
\beqs
\sum_{p\in \mathscr{I}_{j+1} } l_p'c_q = \sum_{p\in \mathscr{I}_{j+1}} \delta_{pq}=1.
\eeqs
Then $l_{j+1}'c_q=0$.
If $l_{j+1}c_q=0$, then $q\notin \mathscr{I}_{j+1}$ and
$\sum_{p\in \mathscr{I}_{j+1}} l_p'c_q = 0$. Again, $l_{j+1}'c_q=0$. When $q=j+1$,
$l'_{j+1}c_{j+1} = l_{j+1}c_{j+1} = 1$ by Eq.~(\ref{eq:lc}).
For $q>j+1$, since $L^{(j+1)}_2$ and $L^{(j+1)}_3$ are both lower triangular, $l_{j+1}'c_q = 0$. Thus, $l'_jc_p=\delta_{jp}$ holds for $1\leq j,p\leq n$.

\end{proof}



\begin{thebibliography}{50}%
\makeatletter
\providecommand \@ifxundefined [1]{%
 \@ifx{#1\undefined}
}%
\providecommand \@ifnum [1]{%
 \ifnum #1\expandafter \@firstoftwo
 \else \expandafter \@secondoftwo
 \fi
}%
\providecommand \@ifx [1]{%
 \ifx #1\expandafter \@firstoftwo
 \else \expandafter \@secondoftwo
 \fi
}%
\providecommand \natexlab [1]{#1}%
\providecommand \enquote  [1]{``#1''}%
\providecommand \bibnamefont  [1]{#1}%
\providecommand \bibfnamefont [1]{#1}%
\providecommand \citenamefont [1]{#1}%
\providecommand \href@noop [0]{\@secondoftwo}%
\providecommand \href [0]{\begingroup \@sanitize@url \@href}%
\providecommand \@href[1]{\@@startlink{#1}\@@href}%
\providecommand \@@href[1]{\endgroup#1\@@endlink}%
\providecommand \@sanitize@url [0]{\catcode `\\12\catcode `\$12\catcode
  `\&12\catcode `\#12\catcode `\^12\catcode `\_12\catcode `\%12\relax}%
\providecommand \@@startlink[1]{}%
\providecommand \@@endlink[0]{}%
\providecommand \url  [0]{\begingroup\@sanitize@url \@url }%
\providecommand \@url [1]{\endgroup\@href {#1}{\urlprefix }}%
\providecommand \urlprefix  [0]{URL }%
\providecommand \Eprint [0]{\href }%
\providecommand \doibase [0]{http://dx.doi.org/}%
\providecommand \selectlanguage [0]{\@gobble}%
\providecommand \bibinfo  [0]{\@secondoftwo}%
\providecommand \bibfield  [0]{\@secondoftwo}%
\providecommand \translation [1]{[#1]}%
\providecommand \BibitemOpen [0]{}%
\providecommand \bibitemStop [0]{}%
\providecommand \bibitemNoStop [0]{.\EOS\space}%
\providecommand \EOS [0]{\spacefactor3000\relax}%
\providecommand \BibitemShut  [1]{\csname bibitem#1\endcsname}%
\let\auto@bib@innerbib\@empty
\bibitem [{\citenamefont {Shor}(1995)}]{Shor:1995:R2493}%
  \BibitemOpen
  \bibfield  {author} {\bibinfo {author} {\bibfnamefont {P.~W.}\ \bibnamefont
  {Shor}},\ }\href@noop {} {\bibfield  {journal} {\bibinfo  {journal} {Phys.
  Rev. A}\ }\textbf {\bibinfo {volume} {52}},\ \bibinfo {pages} {R2493}
  (\bibinfo {year} {1995})}\BibitemShut {NoStop}%
\bibitem [{\citenamefont {Steane}(1996)}]{Steane:1996:793}%
  \BibitemOpen
  \bibfield  {author} {\bibinfo {author} {\bibfnamefont {A.~M.}\ \bibnamefont
  {Steane}},\ }\href@noop {} {\bibfield  {journal} {\bibinfo  {journal} {Phys.
  Rev. Lett.}\ }\textbf {\bibinfo {volume} {77}},\ \bibinfo {pages} {793}
  (\bibinfo {year} {1996})}\BibitemShut {NoStop}%
\bibitem [{\citenamefont {Calderbank}\ and\ \citenamefont
  {Shor}(1996)}]{Calderbank:1996:1098}%
  \BibitemOpen
  \bibfield  {author} {\bibinfo {author} {\bibfnamefont {A.~R.}\ \bibnamefont
  {Calderbank}}\ and\ \bibinfo {author} {\bibfnamefont {P.~W.}\ \bibnamefont
  {Shor}},\ }\href@noop {} {\bibfield  {journal} {\bibinfo  {journal} {Phys.
  Rev. A}\ }\textbf {\bibinfo {volume} {54}},\ \bibinfo {pages} {1098}
  (\bibinfo {year} {1996})}\BibitemShut {NoStop}%
\bibitem [{\citenamefont {Gaitan}(2008)}]{Gaitan:2008:CRC}%
  \BibitemOpen
  \bibfield  {author} {\bibinfo {author} {\bibfnamefont {F.}~\bibnamefont
  {Gaitan}},\ }\href@noop {} {\emph {\bibinfo {title} {Quantum Error Correction
  and Fault Tolerant Quantum Computing}}}\ (\bibinfo  {publisher} {CRC},\
  \bibinfo {year} {2008})\BibitemShut {NoStop}%
\bibitem [{\citenamefont {Lidar}\ and\ \citenamefont
  {Brun}(2013)}]{QECbook:2013}%
  \BibitemOpen
  \bibfield  {author} {\bibinfo {author} {\bibfnamefont {D.}~\bibnamefont
  {Lidar}}\ and\ \bibinfo {author} {\bibfnamefont {T.}~\bibnamefont {Brun}},\
  }\href@noop {} {\emph {\bibinfo {title} {Quantum Error Correction}}}\
  (\bibinfo  {publisher} {Cambridge University Press, Cambridge},\ \bibinfo
  {year} {2013})\BibitemShut {NoStop}%
\bibitem [{\citenamefont {Shor}(1996)}]{Shor:1996:56}%
  \BibitemOpen
  \bibfield  {author} {\bibinfo {author} {\bibfnamefont {P.}~\bibnamefont
  {Shor}},\ }in\ \href@noop {} {\emph {\bibinfo {booktitle} {Proc. 37$^{th}$
  Annual Symposium on Foundations of Computer Science}}}\ (\bibinfo
  {publisher} {IEEE Computer Society Press},\ \bibinfo {address} {Los Alamitos,
  CA},\ \bibinfo {year} {1996})\ p.~\bibinfo {pages} {56}\BibitemShut {NoStop}%
\bibitem [{\citenamefont {Aharonov}\ and\ \citenamefont
  {Ben-Or}(1997)}]{Aharonov:1997:176}%
  \BibitemOpen
  \bibfield  {author} {\bibinfo {author} {\bibfnamefont {D.}~\bibnamefont
  {Aharonov}}\ and\ \bibinfo {author} {\bibfnamefont {M.}~\bibnamefont
  {Ben-Or}},\ }in\ \href@noop {} {\emph {\bibinfo {booktitle} {Proc. 29$^{th}$
  Annual ACM Symposium on the Theory of Computation}}}\ (\bibinfo  {publisher}
  {ACM Press},\ \bibinfo {address} {New York},\ \bibinfo {year} {1997})\ p.\
  \bibinfo {pages} {176}\BibitemShut {NoStop}%
\bibitem [{\citenamefont {{Gottesman}}(1997)}]{Gottesman:9705052}%
  \BibitemOpen
  \bibfield  {author} {\bibinfo {author} {\bibfnamefont {D.}~\bibnamefont
  {{Gottesman}}},\ }\emph {\bibinfo {title} {Stabilizer codes and quantum error
  correction}},\ \href@noop {} {Ph.D. thesis},\ \bibinfo  {school} {California
  Institute of Technology} (\bibinfo {year} {1997}),\ \bibinfo {note} {eprint
  arXiv:quant-ph/9705052}\BibitemShut {NoStop}%
\bibitem [{\citenamefont {Kitaev}(2003)}]{Kitaev:2003:2}%
  \BibitemOpen
  \bibfield  {author} {\bibinfo {author} {\bibfnamefont {A.}~\bibnamefont
  {Kitaev}},\ }\href@noop {} {\bibfield  {journal} {\bibinfo  {journal} {Ann.
  of Phys.}\ }\textbf {\bibinfo {volume} {303}},\ \bibinfo {pages} {2}
  (\bibinfo {year} {2003})}\BibitemShut {NoStop}%
\bibitem [{\citenamefont {DiVincenzo}\ and\ \citenamefont
  {Shor}(1996)}]{DivencenzoFTPhysRevLett.77.3260}%
  \BibitemOpen
  \bibfield  {author} {\bibinfo {author} {\bibfnamefont {D.~P.}\ \bibnamefont
  {DiVincenzo}}\ and\ \bibinfo {author} {\bibfnamefont {P.~W.}\ \bibnamefont
  {Shor}},\ }\href@noop {} {\bibfield  {journal} {\bibinfo  {journal} {Phys.
  Rev. Lett.}\ }\textbf {\bibinfo {volume} {77}},\ \bibinfo {pages} {3260}
  (\bibinfo {year} {1996})}\BibitemShut {NoStop}%
\bibitem [{\citenamefont {Knill}(2005)}]{KnillFTNature}%
  \BibitemOpen
  \bibfield  {author} {\bibinfo {author} {\bibfnamefont {E.}~\bibnamefont
  {Knill}},\ }\href@noop {} {\bibfield  {journal} {\bibinfo  {journal} {Nature
  (London)}\ }\textbf {\bibinfo {volume} {434}},\ \bibinfo {pages} {39}
  (\bibinfo {year} {2005})}\BibitemShut {NoStop}%
\bibitem [{\citenamefont {Aharonov}\ \emph {et~al.}(2006)\citenamefont
  {Aharonov}, \citenamefont {Kitaev},\ and\ \citenamefont
  {Preskill}}]{Aharonov:2006:050504}%
  \BibitemOpen
  \bibfield  {author} {\bibinfo {author} {\bibfnamefont {D.}~\bibnamefont
  {Aharonov}}, \bibinfo {author} {\bibfnamefont {A.}~\bibnamefont {Kitaev}}, \
  and\ \bibinfo {author} {\bibfnamefont {J.}~\bibnamefont {Preskill}},\
  }\href@noop {} {\bibfield  {journal} {\bibinfo  {journal} {Phys. Rev. Lett.}\
  }\textbf {\bibinfo {volume} {96}},\ \bibinfo {pages} {050504} (\bibinfo
  {year} {2006})}\BibitemShut {NoStop}%
\bibitem [{\citenamefont {Terhal}\ and\ \citenamefont
  {Burkard}(2005)}]{Terhal:2005:012336}%
  \BibitemOpen
  \bibfield  {author} {\bibinfo {author} {\bibfnamefont {B.~M.}\ \bibnamefont
  {Terhal}}\ and\ \bibinfo {author} {\bibfnamefont {G.}~\bibnamefont
  {Burkard}},\ }\href@noop {} {\bibfield  {journal} {\bibinfo  {journal} {Phys.
  Rev. A}\ }\textbf {\bibinfo {volume} {71}},\ \bibinfo {pages} {012336}
  (\bibinfo {year} {2005})}\BibitemShut {NoStop}%
\bibitem [{\citenamefont {Aliferis}\ \emph {et~al.}(2006)\citenamefont
  {Aliferis}, \citenamefont {Gottesman},\ and\ \citenamefont
  {Preskill}}]{Aliferis:2006:97}%
  \BibitemOpen
  \bibfield  {author} {\bibinfo {author} {\bibfnamefont {P.}~\bibnamefont
  {Aliferis}}, \bibinfo {author} {\bibfnamefont {D.}~\bibnamefont {Gottesman}},
  \ and\ \bibinfo {author} {\bibfnamefont {J.}~\bibnamefont {Preskill}},\
  }\href@noop {} {\bibfield  {journal} {\bibinfo  {journal} {Quantum Inf.
  Comput.}\ }\textbf {\bibinfo {volume} {6}},\ \bibinfo {pages} {97} (\bibinfo
  {year} {2006})}\BibitemShut {NoStop}%
\bibitem [{\citenamefont {Cross}\ \emph {et~al.}(2009)\citenamefont {Cross},
  \citenamefont {Divincenzo},\ and\ \citenamefont
  {Terhal}}]{cross2007comparative}%
  \BibitemOpen
  \bibfield  {author} {\bibinfo {author} {\bibfnamefont {A.~W.}\ \bibnamefont
  {Cross}}, \bibinfo {author} {\bibfnamefont {D.~P.}\ \bibnamefont
  {Divincenzo}}, \ and\ \bibinfo {author} {\bibfnamefont {B.~M.}\ \bibnamefont
  {Terhal}},\ }\href@noop {} {\bibfield  {journal} {\bibinfo  {journal}
  {Quantum Inf. Comput.}\ }\textbf {\bibinfo {volume} {9}},\ \bibinfo {pages}
  {0541} (\bibinfo {year} {2009})}\BibitemShut {NoStop}%
\bibitem [{\citenamefont {Aliferis}\ \emph {et~al.}(2008)\citenamefont
  {Aliferis}, \citenamefont {Gottesman},\ and\ \citenamefont
  {Preskill}}]{Aliferis:2008:181}%
  \BibitemOpen
  \bibfield  {author} {\bibinfo {author} {\bibfnamefont {P.}~\bibnamefont
  {Aliferis}}, \bibinfo {author} {\bibfnamefont {D.}~\bibnamefont {Gottesman}},
  \ and\ \bibinfo {author} {\bibfnamefont {J.}~\bibnamefont {Preskill}},\
  }\href@noop {} {\bibfield  {journal} {\bibinfo  {journal} {Quantum Inf.
  Comput.}\ }\textbf {\bibinfo {volume} {8}},\ \bibinfo {pages} {181} (\bibinfo
  {year} {2008})}\BibitemShut {NoStop}%
\bibitem [{\citenamefont {Fowler}\ \emph {et~al.}(2012)\citenamefont {Fowler},
  \citenamefont {Mariantoni}, \citenamefont {Martinis},\ and\ \citenamefont
  {Cleland}}]{Folwer2012PhysRevA.86.032324}%
  \BibitemOpen
  \bibfield  {author} {\bibinfo {author} {\bibfnamefont {A.~G.}\ \bibnamefont
  {Fowler}}, \bibinfo {author} {\bibfnamefont {M.}~\bibnamefont {Mariantoni}},
  \bibinfo {author} {\bibfnamefont {J.~M.}\ \bibnamefont {Martinis}}, \ and\
  \bibinfo {author} {\bibfnamefont {A.~N.}\ \bibnamefont {Cleland}},\
  }\href@noop {} {\bibfield  {journal} {\bibinfo  {journal} {Phys. Rev. A}\
  }\textbf {\bibinfo {volume} {86}},\ \bibinfo {pages} {032324} (\bibinfo
  {year} {2012})}\BibitemShut {NoStop}%
\bibitem [{\citenamefont {Bombin}\ and\ \citenamefont
  {Martin-Delgado}(2006)}]{Bombin:2006:180501}%
  \BibitemOpen
  \bibfield  {author} {\bibinfo {author} {\bibfnamefont {H.}~\bibnamefont
  {Bombin}}\ and\ \bibinfo {author} {\bibfnamefont {M.}~\bibnamefont
  {Martin-Delgado}},\ }\href@noop {} {\bibfield  {journal} {\bibinfo  {journal}
  {Phys. Rev. Lett.}\ }\textbf {\bibinfo {volume} {97}},\ \bibinfo {pages}
  {180501} (\bibinfo {year} {2006})}\BibitemShut {NoStop}%
\bibitem [{\citenamefont {Steane}(1999)}]{steane1999efficient_Nature}%
  \BibitemOpen
  \bibfield  {author} {\bibinfo {author} {\bibfnamefont {A.~M.}\ \bibnamefont
  {Steane}},\ }\href@noop {} {\bibfield  {journal} {\bibinfo  {journal}
  {Nature}\ }\textbf {\bibinfo {volume} {399}},\ \bibinfo {pages} {124}
  (\bibinfo {year} {1999})}\BibitemShut {NoStop}%
\bibitem [{\citenamefont {Steane}\ and\ \citenamefont
  {Ibinson}(2005)}]{steane2005fault}%
  \BibitemOpen
  \bibfield  {author} {\bibinfo {author} {\bibfnamefont {A.~M.}\ \bibnamefont
  {Steane}}\ and\ \bibinfo {author} {\bibfnamefont {B.}~\bibnamefont
  {Ibinson}},\ }\href@noop {} {\bibfield  {journal} {\bibinfo  {journal} {Phy.
  Rev. A}\ }\textbf {\bibinfo {volume} {72}},\ \bibinfo {pages} {052335}
  (\bibinfo {year} {2005})}\BibitemShut {NoStop}%
\bibitem [{\citenamefont {Brun}\ \emph {et~al.}(2015)\citenamefont {Brun},
  \citenamefont {Zheng}, \citenamefont {Hsu}, \citenamefont {Job},\ and\
  \citenamefont {Lai}}]{brun2015teleportation}%
  \BibitemOpen
  \bibfield  {author} {\bibinfo {author} {\bibfnamefont {T.~A.}\ \bibnamefont
  {Brun}}, \bibinfo {author} {\bibfnamefont {Y.-C.}\ \bibnamefont {Zheng}},
  \bibinfo {author} {\bibfnamefont {K.-C.}\ \bibnamefont {Hsu}}, \bibinfo
  {author} {\bibfnamefont {J.}~\bibnamefont {Job}}, \ and\ \bibinfo {author}
  {\bibfnamefont {C.-Y.}\ \bibnamefont {Lai}},\ }\href@noop {} {\bibfield
  {journal} {\bibinfo  {journal} {arXiv preprint arXiv:1504.03913}\ } (\bibinfo
  {year} {2015})}\BibitemShut {NoStop}%
\bibitem [{\citenamefont {Steane}(2003)}]{Steane:2003:042322}%
  \BibitemOpen
  \bibfield  {author} {\bibinfo {author} {\bibfnamefont {A.~M.}\ \bibnamefont
  {Steane}},\ }\href@noop {} {\bibfield  {journal} {\bibinfo  {journal} {Phys.
  Rev. A}\ }\textbf {\bibinfo {volume} {68}},\ \bibinfo {pages} {042322}
  (\bibinfo {year} {2003})}\BibitemShut {NoStop}%
\bibitem [{\citenamefont {Gottesman}(2014)}]{gottesman2013overhead}%
  \BibitemOpen
  \bibfield  {author} {\bibinfo {author} {\bibfnamefont {D.}~\bibnamefont
  {Gottesman}},\ }\href@noop {} {\bibfield  {journal} {\bibinfo  {journal}
  {Quantum Inf. Comput.}\ }\textbf {\bibinfo {volume} {14}},\ \bibinfo {pages}
  {1338} (\bibinfo {year} {2014})}\BibitemShut {NoStop}%
\bibitem [{\citenamefont {Steane}(1997)}]{steane1997active}%
  \BibitemOpen
  \bibfield  {author} {\bibinfo {author} {\bibfnamefont {A.~M.}\ \bibnamefont
  {Steane}},\ }\href@noop {} {\bibfield  {journal} {\bibinfo  {journal} {Phys.
  Rev. Lett.}\ }\textbf {\bibinfo {volume} {78}},\ \bibinfo {pages} {2252}
  (\bibinfo {year} {1997})}\BibitemShut {NoStop}%
\bibitem [{\citenamefont {Gottesman}\ and\ \citenamefont
  {Chuang}(1999)}]{Gottesman:1999:390}%
  \BibitemOpen
  \bibfield  {author} {\bibinfo {author} {\bibfnamefont {D.}~\bibnamefont
  {Gottesman}}\ and\ \bibinfo {author} {\bibfnamefont {I.}~\bibnamefont
  {Chuang}},\ }\href@noop {} {\bibfield  {journal} {\bibinfo  {journal}
  {Nature}\ }\textbf {\bibinfo {volume} {402}},\ \bibinfo {pages} {390}
  (\bibinfo {year} {1999})}\BibitemShut {NoStop}%
\bibitem [{\citenamefont {Zhou}\ \emph {et~al.}(2000)\citenamefont {Zhou},
  \citenamefont {Leung},\ and\ \citenamefont {Chuang}}]{Zhou:2000:052316}%
  \BibitemOpen
  \bibfield  {author} {\bibinfo {author} {\bibfnamefont {X.}~\bibnamefont
  {Zhou}}, \bibinfo {author} {\bibfnamefont {D.~W.}\ \bibnamefont {Leung}}, \
  and\ \bibinfo {author} {\bibfnamefont {I.~L.}\ \bibnamefont {Chuang}},\
  }\href@noop {} {\bibfield  {journal} {\bibinfo  {journal} {Phys. Rev. A}\
  }\textbf {\bibinfo {volume} {62}},\ \bibinfo {pages} {052316} (\bibinfo
  {year} {2000})}\BibitemShut {NoStop}%
\bibitem [{\citenamefont {Patel}\ \emph {et~al.}(2008)\citenamefont {Patel},
  \citenamefont {Markov},\ and\ \citenamefont {Hayes}}]{markov2008optimal}%
  \BibitemOpen
  \bibfield  {author} {\bibinfo {author} {\bibfnamefont {K.~N.}\ \bibnamefont
  {Patel}}, \bibinfo {author} {\bibfnamefont {I.~L.}\ \bibnamefont {Markov}}, \
  and\ \bibinfo {author} {\bibfnamefont {J.~P.}\ \bibnamefont {Hayes}},\
  }\href@noop {} {\bibfield  {journal} {\bibinfo  {journal} {Quantum Inf.
  Comput.}\ }\textbf {\bibinfo {volume} {8}},\ \bibinfo {pages} {0282}
  (\bibinfo {year} {2008})}\BibitemShut {NoStop}%
\bibitem [{\citenamefont {Aaronson}\ and\ \citenamefont
  {Gottesman}(2004)}]{aaronson2004improved}%
  \BibitemOpen
  \bibfield  {author} {\bibinfo {author} {\bibfnamefont {S.}~\bibnamefont
  {Aaronson}}\ and\ \bibinfo {author} {\bibfnamefont {D.}~\bibnamefont
  {Gottesman}},\ }\href@noop {} {\bibfield  {journal} {\bibinfo  {journal}
  {Phys. Rev. A}\ }\textbf {\bibinfo {volume} {70}},\ \bibinfo {pages} {052328}
  (\bibinfo {year} {2004})}\BibitemShut {NoStop}%
\bibitem [{\citenamefont {Lai}\ \emph {et~al.}(2017)\citenamefont {Lai},
  \citenamefont {Zheng},\ and\ \citenamefont {Brun}}]{Ancilla_distillation_1}%
  \BibitemOpen
  \bibfield  {author} {\bibinfo {author} {\bibfnamefont {C.-Y.}\ \bibnamefont
  {Lai}}, \bibinfo {author} {\bibfnamefont {Y.-C.}\ \bibnamefont {Zheng}}, \
  and\ \bibinfo {author} {\bibfnamefont {T.~A.}\ \bibnamefont {Brun}},\
  }\href@noop {} {\bibfield  {journal} {\bibinfo  {journal} {Phys. Rev. A}\
  }\textbf {\bibinfo {volume} {95}},\ \bibinfo {pages} {032339} (\bibinfo
  {year} {2017})}\BibitemShut {NoStop}%
\bibitem [{\citenamefont {Zheng}\ \emph {et~al.}(2018)\citenamefont {Zheng},
  \citenamefont {Lai},\ and\ \citenamefont {Brun}}]{zheng2017efficient}%
  \BibitemOpen
  \bibfield  {author} {\bibinfo {author} {\bibfnamefont {Y.-C.}\ \bibnamefont
  {Zheng}}, \bibinfo {author} {\bibfnamefont {C.-Y.}\ \bibnamefont {Lai}}, \
  and\ \bibinfo {author} {\bibfnamefont {T.~A.}\ \bibnamefont {Brun}},\
  }\href@noop {} {\bibfield  {journal} {\bibinfo  {journal} {Phys. Rev. A}\
  }\textbf {\bibinfo {volume} {97}},\ \bibinfo {pages} {032331} (\bibinfo
  {year} {2018})}\BibitemShut {NoStop}%
\bibitem [{Note1()}]{Note1}%
  \BibitemOpen
  \bibinfo {note} {Such extra layer has depth $O(1)$. Throughout the paper,
  Pauli gates are assumed to be free and can be directly applied to qubits.
  This is also true in FTQC using stabilizer codes, where logical Pauli
  operators are easy to realize.}\BibitemShut {Stop}%
\bibitem [{\citenamefont {Maslov}\ and\ \citenamefont
  {Roetteler}(2018)}]{maslov2017Bruhat}%
  \BibitemOpen
  \bibfield  {author} {\bibinfo {author} {\bibfnamefont {D.}~\bibnamefont
  {Maslov}}\ and\ \bibinfo {author} {\bibfnamefont {M.}~\bibnamefont
  {Roetteler}},\ }\href@noop {} {\bibfield  {journal} {\bibinfo  {journal}
  {IEEE Trans. Inf. Theory}\ }\textbf {\bibinfo {volume} {64}},\ \bibinfo
  {pages} {4729} (\bibinfo {year} {2018})}\BibitemShut {NoStop}%
\bibitem [{\citenamefont {Chamberland}\ and\ \citenamefont
  {Ronagh}(2018)}]{chamberland2018deep}%
  \BibitemOpen
  \bibfield  {author} {\bibinfo {author} {\bibfnamefont {C.}~\bibnamefont
  {Chamberland}}\ and\ \bibinfo {author} {\bibfnamefont {P.}~\bibnamefont
  {Ronagh}},\ }\href@noop {} {\bibfield  {journal} {\bibinfo  {journal}
  {Quantum Sci.Tech.}\ }\textbf {\bibinfo {volume} {3}},\ \bibinfo {pages}
  {044002} (\bibinfo {year} {2018})}\BibitemShut {NoStop}%
\bibitem [{\citenamefont {Steane}(2002)}]{steane2002fast}%
  \BibitemOpen
  \bibfield  {author} {\bibinfo {author} {\bibfnamefont {A.~M.}\ \bibnamefont
  {Steane}},\ }\href@noop {} {\bibfield  {journal} {\bibinfo  {journal} {arXiv
  preprint quant-ph/0202036}\ } (\bibinfo {year} {2002})}\BibitemShut {NoStop}%
\bibitem [{\citenamefont
  {MacKay}(2003)}]{MacKay:2003:CambridgeUniversityPress}%
  \BibitemOpen
  \bibfield  {author} {\bibinfo {author} {\bibfnamefont {D.~J.~C.}\
  \bibnamefont {MacKay}},\ }\href@noop {} {\emph {\bibinfo {title} {Information
  Theory, Inference and Learning Algorithms}}}\ (\bibinfo  {publisher}
  {Cambridge University Press},\ \bibinfo {address} {Cambridge, UK},\ \bibinfo
  {year} {2003})\BibitemShut {NoStop}%
\bibitem [{\citenamefont {Lloyd}(1996)}]{Lloyd:1996:1073}%
  \BibitemOpen
  \bibfield  {author} {\bibinfo {author} {\bibfnamefont {S.}~\bibnamefont
  {Lloyd}},\ }\href@noop {} {\bibfield  {journal} {\bibinfo  {journal}
  {Science}\ }\textbf {\bibinfo {volume} {273}},\ \bibinfo {pages} {1073}
  (\bibinfo {year} {1996})}\BibitemShut {NoStop}%
\bibitem [{\citenamefont {Aspuru-Guzik}\ \emph {et~al.}(2005)\citenamefont
  {Aspuru-Guzik}, \citenamefont {Dutoi}, \citenamefont {Love},\ and\
  \citenamefont {Head-Gordon}}]{aspuru2005simulated}%
  \BibitemOpen
  \bibfield  {author} {\bibinfo {author} {\bibfnamefont {A.}~\bibnamefont
  {Aspuru-Guzik}}, \bibinfo {author} {\bibfnamefont {A.~D.}\ \bibnamefont
  {Dutoi}}, \bibinfo {author} {\bibfnamefont {P.~J.}\ \bibnamefont {Love}}, \
  and\ \bibinfo {author} {\bibfnamefont {M.}~\bibnamefont {Head-Gordon}},\
  }\href {\doibase 10.1126/science.1113479} {\bibfield  {journal} {\bibinfo
  {journal} {Science}\ }\textbf {\bibinfo {volume} {309}},\ \bibinfo {pages}
  {1704} (\bibinfo {year} {2005})}\BibitemShut {NoStop}%
\bibitem [{\citenamefont {Wecker}\ \emph {et~al.}(2014)\citenamefont {Wecker},
  \citenamefont {Bauer}, \citenamefont {Clark}, \citenamefont {Hastings},\ and\
  \citenamefont {Troyer}}]{wecker2014gate}%
  \BibitemOpen
  \bibfield  {author} {\bibinfo {author} {\bibfnamefont {D.}~\bibnamefont
  {Wecker}}, \bibinfo {author} {\bibfnamefont {B.}~\bibnamefont {Bauer}},
  \bibinfo {author} {\bibfnamefont {B.~K.}\ \bibnamefont {Clark}}, \bibinfo
  {author} {\bibfnamefont {M.~B.}\ \bibnamefont {Hastings}}, \ and\ \bibinfo
  {author} {\bibfnamefont {M.}~\bibnamefont {Troyer}},\ }\href@noop {}
  {\bibfield  {journal} {\bibinfo  {journal} {Phys. Rev. A}\ }\textbf {\bibinfo
  {volume} {90}},\ \bibinfo {pages} {022305} (\bibinfo {year}
  {2014})}\BibitemShut {NoStop}%
\bibitem [{\citenamefont {Hastings}\ \emph {et~al.}(2015)\citenamefont
  {Hastings}, \citenamefont {Wecker}, \citenamefont {Bauer},\ and\
  \citenamefont {Troyer}}]{Hastings:2015_simulation}%
  \BibitemOpen
  \bibfield  {author} {\bibinfo {author} {\bibfnamefont {M.~B.}\ \bibnamefont
  {Hastings}}, \bibinfo {author} {\bibfnamefont {D.}~\bibnamefont {Wecker}},
  \bibinfo {author} {\bibfnamefont {B.}~\bibnamefont {Bauer}}, \ and\ \bibinfo
  {author} {\bibfnamefont {M.}~\bibnamefont {Troyer}},\ }\href@noop {}
  {\bibfield  {journal} {\bibinfo  {journal} {Quantum Inf. Comput.}\ }\textbf
  {\bibinfo {volume} {15}},\ \bibinfo {pages} {1} (\bibinfo {year}
  {2015})}\BibitemShut {NoStop}%
\bibitem [{\citenamefont {Poulin}\ \emph {et~al.}(2015)\citenamefont {Poulin},
  \citenamefont {Hastings}, \citenamefont {Wecker}, \citenamefont {Wiebe},
  \citenamefont {Doberty},\ and\ \citenamefont
  {Troyer}}]{Poulin:2015qic_simulation}%
  \BibitemOpen
  \bibfield  {author} {\bibinfo {author} {\bibfnamefont {D.}~\bibnamefont
  {Poulin}}, \bibinfo {author} {\bibfnamefont {M.~B.}\ \bibnamefont
  {Hastings}}, \bibinfo {author} {\bibfnamefont {D.}~\bibnamefont {Wecker}},
  \bibinfo {author} {\bibfnamefont {N.}~\bibnamefont {Wiebe}}, \bibinfo
  {author} {\bibfnamefont {A.~C.}\ \bibnamefont {Doberty}}, \ and\ \bibinfo
  {author} {\bibfnamefont {M.}~\bibnamefont {Troyer}},\ }\href@noop {}
  {\bibfield  {journal} {\bibinfo  {journal} {Quantum Inf. Comput.}\ }\textbf
  {\bibinfo {volume} {15}},\ \bibinfo {pages} {361} (\bibinfo {year}
  {2015})}\BibitemShut {NoStop}%
\bibitem [{\citenamefont {Bauer}\ \emph {et~al.}(2020)\citenamefont {Bauer},
  \citenamefont {Bravyi}, \citenamefont {Motta},\ and\ \citenamefont
  {Chan}}]{garnet2020quantum}%
  \BibitemOpen
  \bibfield  {author} {\bibinfo {author} {\bibfnamefont {B.}~\bibnamefont
  {Bauer}}, \bibinfo {author} {\bibfnamefont {S.}~\bibnamefont {Bravyi}},
  \bibinfo {author} {\bibfnamefont {M.}~\bibnamefont {Motta}}, \ and\ \bibinfo
  {author} {\bibfnamefont {G.~K.}\ \bibnamefont {Chan}},\ }\href@noop {}
  {\bibfield  {journal} {\bibinfo  {journal} {arXiv preprint arXiv:2001.03685}\
  } (\bibinfo {year} {2020})}\BibitemShut {NoStop}%
\bibitem [{\citenamefont {Farhi}\ \emph {et~al.}(2014)\citenamefont {Farhi},
  \citenamefont {Goldstone},\ and\ \citenamefont
  {Gutmann}}]{farhi2014quantumqaoa}%
  \BibitemOpen
  \bibfield  {author} {\bibinfo {author} {\bibfnamefont {E.}~\bibnamefont
  {Farhi}}, \bibinfo {author} {\bibfnamefont {J.}~\bibnamefont {Goldstone}}, \
  and\ \bibinfo {author} {\bibfnamefont {S.}~\bibnamefont {Gutmann}},\
  }\href@noop {} {\bibfield  {journal} {\bibinfo  {journal} {arXiv preprint
  arXiv:1411.4028}\ } (\bibinfo {year} {2014})}\BibitemShut {NoStop}%
\bibitem [{\citenamefont {Heckey}\ \emph {et~al.}(2015)\citenamefont {Heckey},
  \citenamefont {Patil}, \citenamefont {JavadiAbhari}, \citenamefont {Holmes},
  \citenamefont {Kudrow}, \citenamefont {Brown}, \citenamefont {Franklin},
  \citenamefont {Chong},\ and\ \citenamefont {Martonosi}}]{heckey2015compiler}%
  \BibitemOpen
  \bibfield  {author} {\bibinfo {author} {\bibfnamefont {J.}~\bibnamefont
  {Heckey}}, \bibinfo {author} {\bibfnamefont {S.}~\bibnamefont {Patil}},
  \bibinfo {author} {\bibfnamefont {A.}~\bibnamefont {JavadiAbhari}}, \bibinfo
  {author} {\bibfnamefont {A.}~\bibnamefont {Holmes}}, \bibinfo {author}
  {\bibfnamefont {D.}~\bibnamefont {Kudrow}}, \bibinfo {author} {\bibfnamefont
  {K.~R.}\ \bibnamefont {Brown}}, \bibinfo {author} {\bibfnamefont
  {D.}~\bibnamefont {Franklin}}, \bibinfo {author} {\bibfnamefont {F.~T.}\
  \bibnamefont {Chong}}, \ and\ \bibinfo {author} {\bibfnamefont
  {M.}~\bibnamefont {Martonosi}},\ }in\ \href@noop {} {\emph {\bibinfo
  {booktitle} {Proc. 29$^{th}$ International Conference on Architectural
  Support for Programming Languages and Operating Systems (ASPLOS)}}}\
  (\bibinfo {organization} {ACM Press, Istanbul, Turkey},\ \bibinfo {year}
  {2015})\ pp.\ \bibinfo {pages} {445--456}\BibitemShut {NoStop}%
\bibitem [{\citenamefont {Risque}\ and\ \citenamefont
  {Jog}(2016)}]{risque2016characterization}%
  \BibitemOpen
  \bibfield  {author} {\bibinfo {author} {\bibfnamefont {R.}~\bibnamefont
  {Risque}}\ and\ \bibinfo {author} {\bibfnamefont {A.}~\bibnamefont {Jog}},\
  }in\ \href@noop {} {\emph {\bibinfo {booktitle} {IEEE International Symposium
  on Workload Characterization (IISWC)}}}\ (\bibinfo {organization} {IEEE
  Computer Society Press, Providence, RI},\ \bibinfo {year} {2016})\ pp.\
  \bibinfo {pages} {1--9}\BibitemShut {NoStop}%
\bibitem [{\citenamefont {Raussendorf}\ \emph {et~al.}(2003)\citenamefont
  {Raussendorf}, \citenamefont {Browne},\ and\ \citenamefont
  {Briegel}}]{Raussendorf:2003:022312}%
  \BibitemOpen
  \bibfield  {author} {\bibinfo {author} {\bibfnamefont {R.}~\bibnamefont
  {Raussendorf}}, \bibinfo {author} {\bibfnamefont {D.~E.}\ \bibnamefont
  {Browne}}, \ and\ \bibinfo {author} {\bibfnamefont {H.~J.}\ \bibnamefont
  {Briegel}},\ }\href@noop {} {\bibfield  {journal} {\bibinfo  {journal} {Phys.
  Rev. A}\ }\textbf {\bibinfo {volume} {68}},\ \bibinfo {pages} {022312}
  (\bibinfo {year} {2003})}\BibitemShut {NoStop}%
\bibitem [{\citenamefont {Raussendorf}\ \emph {et~al.}(2006)\citenamefont
  {Raussendorf}, \citenamefont {Harrington},\ and\ \citenamefont
  {Goyal}}]{Raussendorf:2006:2242}%
  \BibitemOpen
  \bibfield  {author} {\bibinfo {author} {\bibfnamefont {R.}~\bibnamefont
  {Raussendorf}}, \bibinfo {author} {\bibfnamefont {J.}~\bibnamefont
  {Harrington}}, \ and\ \bibinfo {author} {\bibfnamefont {K.}~\bibnamefont
  {Goyal}},\ }\href@noop {} {\bibfield  {journal} {\bibinfo  {journal} {Ann.
  Phys.}\ }\textbf {\bibinfo {volume} {321}},\ \bibinfo {pages} {2242}
  (\bibinfo {year} {2006})}\BibitemShut {NoStop}%
\bibitem [{\citenamefont {Raussendorf}\ \emph {et~al.}(2007)\citenamefont
  {Raussendorf}, \citenamefont {Harrington},\ and\ \citenamefont
  {Goyal}}]{Raussendorf:2007:199}%
  \BibitemOpen
  \bibfield  {author} {\bibinfo {author} {\bibfnamefont {R.}~\bibnamefont
  {Raussendorf}}, \bibinfo {author} {\bibfnamefont {J.}~\bibnamefont
  {Harrington}}, \ and\ \bibinfo {author} {\bibfnamefont {K.}~\bibnamefont
  {Goyal}},\ }\href@noop {} {\bibfield  {journal} {\bibinfo  {journal} {New J.
  Phys.}\ }\textbf {\bibinfo {volume} {9}},\ \bibinfo {pages} {199} (\bibinfo
  {year} {2007})}\BibitemShut {NoStop}%
\bibitem [{\citenamefont {Raussendorf}\ \emph {et~al.}(2005)\citenamefont
  {Raussendorf}, \citenamefont {Bravyi},\ and\ \citenamefont
  {Harrington}}]{raussendorf2005long}%
  \BibitemOpen
  \bibfield  {author} {\bibinfo {author} {\bibfnamefont {R.}~\bibnamefont
  {Raussendorf}}, \bibinfo {author} {\bibfnamefont {S.}~\bibnamefont {Bravyi}},
  \ and\ \bibinfo {author} {\bibfnamefont {J.}~\bibnamefont {Harrington}},\
  }\href@noop {} {\bibfield  {journal} {\bibinfo  {journal} {Phys. Rev. A}\
  }\textbf {\bibinfo {volume} {71}},\ \bibinfo {pages} {062313} (\bibinfo
  {year} {2005})}\BibitemShut {NoStop}%
\bibitem [{\citenamefont {Arute}\ \emph {et~al.}(2019)\citenamefont {Arute},
  \citenamefont {Arya}, \citenamefont {Babbush}, \citenamefont {Bacon},
  \citenamefont {Bardin}, \citenamefont {Barends}, \citenamefont {Biswas},
  \citenamefont {Boixo}, \citenamefont {Brandao}, \citenamefont {Buell} \emph
  {et~al.}}]{supremacy}%
  \BibitemOpen
  \bibfield  {author} {\bibinfo {author} {\bibfnamefont {F.}~\bibnamefont
  {Arute}}, \bibinfo {author} {\bibfnamefont {K.}~\bibnamefont {Arya}},
  \bibinfo {author} {\bibfnamefont {R.}~\bibnamefont {Babbush}}, \bibinfo
  {author} {\bibfnamefont {D.}~\bibnamefont {Bacon}}, \bibinfo {author}
  {\bibfnamefont {J.~C.}\ \bibnamefont {Bardin}}, \bibinfo {author}
  {\bibfnamefont {R.}~\bibnamefont {Barends}}, \bibinfo {author} {\bibfnamefont
  {R.}~\bibnamefont {Biswas}}, \bibinfo {author} {\bibfnamefont
  {S.}~\bibnamefont {Boixo}}, \bibinfo {author} {\bibfnamefont {F.~G.}\
  \bibnamefont {Brandao}}, \bibinfo {author} {\bibfnamefont {D.~A.}\
  \bibnamefont {Buell}},  \emph {et~al.},\ }\href@noop {} {\bibfield  {journal}
  {\bibinfo  {journal} {Nature}\ }\textbf {\bibinfo {volume} {574}},\ \bibinfo
  {pages} {505} (\bibinfo {year} {2019})}\BibitemShut {NoStop}%
\bibitem [{\citenamefont {Paetznick}\ and\ \citenamefont
  {Reichardt}(2012)}]{paetznick_ben2011fault}%
  \BibitemOpen
  \bibfield  {author} {\bibinfo {author} {\bibfnamefont {A.}~\bibnamefont
  {Paetznick}}\ and\ \bibinfo {author} {\bibfnamefont {B.~W.}\ \bibnamefont
  {Reichardt}},\ }\href@noop {} {\bibfield  {journal} {\bibinfo  {journal}
  {Quantum Inf. Comput.}\ }\textbf {\bibinfo {volume} {12}},\ \bibinfo {pages}
  {1034} (\bibinfo {year} {2012})}\BibitemShut {NoStop}%
\end{thebibliography}

%

\end{document}